\newif\iflong
\longtrue
\documentclass[acmsmall,screen]{acmart}
\iflong
\setcopyright{none}
\else
\setcopyright{rightsretained}
\fi
\acmJournal{PACMPL}
\acmYear{2020} \acmVolume{4} \acmNumber{POPL} \acmArticle{5} \acmMonth{1} \acmPrice{}\acmDOI{10.1145/3371073}
\copyrightyear{2020}

\iflong
\renewcommand\footnotetextcopyrightpermission[1]{}
\fi

\usepackage{listings}
\lstdefinelanguage{program}{%
  keywords={%
    let,pass,function,%
    var,const,bool,int,void,atomic,%
    while,do,if,then,else,assume,assert,call,return,rule,forall,with,new,choose,skip,%
    task,async,yield,for,wait,%
    type,relation,init, action, safety, invariant, axiom, input
  },
  morecomment=[l]{//},
  morecomment=[s]{/*}{*/},
  morecomment=[n]{(*}{*)},
  mathescape=true,
  escapeinside=`',
}
\usepackage{booktabs}
\usepackage{mathpartir}
\usepackage{multirow}
\usepackage{amsmath}
\usepackage{amsthm}
\usepackage{stmaryrd}
\usepackage{array}
\usepackage{subcaption}
\usepackage[T1]{fontenc}
\usepackage[nodayofweek]{datetime}
\usepackage{graphicx}
\usepackage{algorithm}
\usepackage[noend]{algpseudocode}
\usepackage[flushleft]{threeparttable}
\usepackage{stmaryrd}
\usepackage{multicol}
\usepackage{pifont}
\usepackage[xcolor]{changebar}
\cbcolor{magenta}

\usepackage{xcolor}
\usepackage{extarrows}
\usepackage{mathtools}
\usepackage{paralist}

\usepackage{tikz}
\usetikzlibrary{positioning,fit}

\usepackage[noabbrev,capitalise]{cleveref} %

\newif\ifnitpick
\nitpickfalse

\newif\ifproofs
\proofstrue

\iflong
\newcommand{\refappendix}[1]{\Cref{#1}}
\else
\newcommand{\refappendix}[1]{the extended version~\cite{extendedVersion}}
\fi

\crefformat{section}{\S#2#1#3} %
\crefformat{subsection}{\S#2#1#3}
\crefformat{subsubsection}{\S#2#1#3}

\Crefname{conjecture}{Conjecture}{Conjectures}
\Crefname{proposition}{Proposition}{Propositions}
\Crefname{lemma}{Lemma}{Lemmas}
\Crefname{corollary}{Corollary}{Corollaries}
\Crefname{example}{Example}{Examples}
\Crefname{definition}{Def.}{Defs.}
\Crefname{algorithm}{Alg.}{Alg.}
\Crefname{theorem}{Thm.}{Thm.}
\Crefname{figure}{Fig.}{Fig.}
\crefname{line}{line}{lines}

\ifproofs
\else
\excludecomment{proof}
\fi

\newtheorem{theo}{Theo}[section] %

\newtheorem{remark}[theo]{Remark}
\newcommand{\para}[1]{\vspace{2pt}\noindent\textbf{\textit{#1.}}}

\newcommand{\Nat}{{\mathbb{N}}}

\newcommand{\ov}{\overline}

\newcommand{\A}{\mathcal{A}}
\newcommand{\B}{\mathcal{B}}

\newcommand{\card}[1]{{\left\vert{#1}\right\vert}} %

\renewcommand{\implies}{\Longrightarrow}

\newcommand{\true}{{\textit{true}}}
\newcommand{\false}{{\textit{false}}}
\newcommand{\vocabulary}{\Sigma}
\newcommand{\voc}{\vocabulary}

\newcommand{\Init}{{\textit{Init}}}
\newcommand{\Bad}{\textit{Bad}}

\newcommand{\TS}{{\textit{TS}}}

\newcommand{\tr}{{\delta}}

\newcommand{\sat}{\mbox{{\rm\sc SAT}}}

\newcommand{\Formulas}[1]{\mathcal{F}({#1})}
\newcommand{\cnf}[2]{\text{CNF}_{#1}^{#2}}
\newcommand{\dnf}[2]{\text{DNF}_{#1}^{#2}}
\newcommand{\moncnf}[2]{\text{Mon-}\cnf{#1}{#2}}

\newcommand{\coNP}{\textbf{coNP}}

\newcommand{\fts}[1]{\TS^{{#1}}}
\newcommand{\psigmatr}[1]{\tr_{#1}^{P}}
\newcommand{\hardclass}{\progs{\psigma{2}}}
\newcommand{\tsclass}{\mathcal{P}}
\newcommand{\progs}[1]{\mathcal{P}_{#1}}

\newcommand{\ftset}{\hardclass}
\newcommand{\psigma}[1]{\Sigma_{#1}^{P}}
\newcommand{\ppi}[1]{\Pi_{#1}^{p}}

\renewcommand{\vec}{\ov}
\newcommand{\horaclesym}{\mathcal{H}}
\newcommand{\horacle}[2]{\horaclesym({#1},{#2})}
\newcommand{\exthoracle}[3]{\exthoraclet{\horaclesym}{#1}{#2}{#3}}
\newcommand{\exthoraclet}[4]{{#1}^{({#2})}({#3},{#4})}
\newcommand{\itporaclesym}{\textsc{itp}}
\newcommand{\indclesym}{\mathcal{I}}
\newcommand{\indcle}[2]{\indclet{\indclesym}{#1}{#2}}
\newcommand{\indclet}[3]{{#1}({#2},{#3})}

\newcommand{\set}[1]{\{{#1}\}}
\newcommand{\seq}[2]{\{{#1}\}_{#2}}

\newcommand{\maxtr}[1]{\tr_{#1}^{\mathcal{M}}}
\newcommand{\maxinv}[1]{{#1}}

\newcommand{\bbalg}[2]{\bbalgt{\A}{#1}{#2}}

\newcommand{\bbalgt}[3]{{#1}^{#2}_{#3}}
\newcommand{\bb}[1]{[{#1}]}

\newcommand{\eqdef}{\stackrel{\rm def}{=}}

\newcommand{\qoraclefam}{\mathcal{Q}}

\newcommand{\querycomplexity}{q}

\newcommand{\badts}[1]{{E}_{#1}}
\newcommand{\badtsclass}[1]{\mathcal{E}_{#1}}

\newcommand{\invtswbad}[1]{{\mathcal{R}^{+}}({#1})}

\newcommand{\monmax}{\mathcal{M}_{\badts{}}}

\newcommand{\monmaxerror}{\mathcal{M}_{\badts{}}}
\newcommand{\monmaxnoerror}{\mathcal{M}}

\newcommand{\explowerbound}{2^{\Omega(n)}}

\newcommand{\dual}[1]{{#1}^{*}}

\newcommand{\cube}[1]{\textit{cube}({#1})}
\newcommand{\qbf}[1]{{\rm QBF}_{#1}}
\newcommand{\hoaretrip}[3]{\{{#1}\}{#2}\{{#3}\}}

\newcommand{\prop}{x} 

\newcommand{\vmark}{\ding{51}}%
\newcommand{\xmark}{\ding{55}} 
\begin{document}
\newif\ifcomments
\commentsfalse
\nochangebars
\definecolor{dg}{cmyk}{0.60,0,0.88,0.27}

\ifcomments
\newcommand{\sharon}[1]{{\textcolor{blue}{SS: {\em #1}}}}
\newcommand{\jrw}[1]{{\textcolor{green}{JRW: {\em #1}}}}
\newcommand{\mooly}[1]{{\textcolor{cyan}{MS: {\em #1}}}}
\newcommand{\neil}[1]{{\textcolor{dg}{NI: {\em #1}}}}
\newcommand{\yotam}[1]{{\textcolor{magenta}{YF: {\em #1}}}}
\newcommand{\TODO}[1]{{\textcolor{red}{TODO: {\em #1}}}}

\else
\newcommand{\sharon}[1]{}
\newcommand{\jrw}[1]{}
\newcommand{\mooly}[1]{}
\newcommand{\neil}[1]{}
\newcommand{\yotam}[1]{}
\newcommand{\TODO}[1]{}

\fi

\newcommand{\commentout}[1]{}
\newcommand{\OMIT}[1]{}

\title{Complexity and Information in Invariant Inference}

\author{Yotam M. Y. Feldman}
\affiliation{
  \institution{Tel Aviv University}
  \country{Israel}
}
\email{yotam.feldman@gmail.com}

\author{Neil Immerman}
\affiliation{
  \institution{UMass Amherst}
  \country{USA}
}
\email{immerman@cs.umass.edu}

\author{Mooly Sagiv}
\affiliation{
  \institution{Tel Aviv University}
  \country{Israel}
}
\email{msagiv@acm.org}

\author{Sharon Shoham}
\affiliation{
  \institution{Tel Aviv University}
  \country{Israel}
}
\email{sharon.shoham@gmail.com}

\begin{abstract}
This paper addresses the complexity of SAT-based invariant inference, a prominent approach to safety verification.
We consider the problem of inferring an inductive invariant of \emph{polynomial length} given a transition system and a safety property.
We analyze the complexity of this problem in a black-box model, called \emph{the Hoare-query model}, which is general enough to capture algorithms such as IC3/PDR and its variants.
An algorithm in this model learns about the system's reachable states by querying the validity of Hoare triples.

We show that in general an algorithm in the Hoare-query model requires an exponential number of queries. %
Our lower bound is information-theoretic and applies even to computationally unrestricted algorithms, showing that no choice of generalization from the partial information obtained in a polynomial number of Hoare queries can lead to an efficient invariant inference procedure in this class.

We then show, for the first time, that by utilizing rich Hoare queries, as done in PDR, inference can be exponentially more efficient than approaches such as ICE learning, which only utilize inductiveness checks of candidates.
We do so by constructing a class of transition systems for which a simple version of PDR with a single frame infers invariants in a polynomial number of queries, whereas every algorithm using only inductiveness checks and counterexamples requires an exponential number of queries.

Our results also shed light on connections and differences with the classical theory of exact concept learning with queries, and \begin{changebar}imply that learning from counterexamples to induction is harder than classical exact learning from labeled examples. This demonstrates that the convergence rate of Counterexample-Guided Inductive Synthesis depends on the form of counterexamples.\end{changebar}
\end{abstract}


\maketitle
\iflong
\thispagestyle{empty}
\fi

\section{Introduction}
The inference of inductive invariants is a fundamental technique in safety verification, and the focus of many works~\cite[e.g.][]{DBLP:conf/cav/McMillan03,ic3,pdr,DBLP:conf/popl/CousotC77,DBLP:journals/sttt/SrivastavaGF13,DBLP:series/natosec/AlurBDF0JKMMRSSSSTU15,DBLP:conf/tacas/FedyukovichB18,DBLP:conf/oopsla/DilligDLM13}.
The task is to find an assertion $I$ that holds in the initial states of the system, excludes all bad states, and is closed under transitions of the system, namely, the Hoare triple $\hoaretrip{I}{\tr}{I}$ is valid, where $\tr$ denotes one step of the system. Such an $I$ overapproximates the set of reachable states and establishes their safety.

The advance of SAT-based reasoning has led to the development of successful algorithms inferring inductive invariants using SAT queries. A prominent example is IC3/PDR~\cite{ic3,pdr}, which has led to a significant improvement in the ability to verify realistic hardware systems.
Recently, this algorithm has been extended and generalized to software systems~\cite[e.g.][]{DBLP:conf/sat/HoderB12,DBLP:conf/cav/KomuravelliGC14,DBLP:conf/vmcai/BjornerG15,DBLP:journals/jacm/KarbyshevBIRS17,DBLP:conf/tacas/CimattiGMT14}.

Successful SAT-based inference algorithms are typically tricky and employ many clever heuristics.
This is in line with the inherent asymptotic complexity of invariant inference, which is hard even with access to a SAT solver~\cite{DBLP:conf/cade/LahiriQ09}.
However, the practical success of inference algorithms calls for a more refined complexity analysis, with the objective of understanding the principles on which these algorithms are based. %
This paper studies the asymptotic complexity of SAT-based invariant inference through the decision problem of \emph{polynomial length inference} in the black-box \emph{Hoare-query model}, as we now explain.

\paragraph{\bf Inference of polynomial-length CNF}
Naturally, inference algorithms succeed when the invariant they infer is not too long.
Therefore, this paper considers the complexity of \emph{inferring invariants of polynomial length}. %
We follow the recent trend in invariant inference, advocated in~\cite{DBLP:conf/cav/McMillan03,ic3}, to search for invariants in rich syntactical forms, beyond those usually considered in template-based invariant inference~\cite[e.g.][]{DBLP:conf/popl/JeannetSS14,DBLP:conf/cav/ColonSS03,DBLP:conf/sas/SankaranarayananSM04,DBLP:conf/pldi/SrivastavaG09,DBLP:journals/sttt/SrivastavaGF13,DBLP:series/natosec/AlurBDF0JKMMRSSSSTU15}, with the motivation of achieving generality of the verification method and potentially improving the success rate.
We thus study the inference of invariants expressed in Conjunctive Normal Form (CNF) of polynomial length.
Interestingly, our results also apply to inferring invariants in Disjunctive Normal Form.

\paragraph{\bf The Hoare-query model}
Our study of SAT-based methods focuses on an algorithmic model called the \emph{Hoare-query model}.
The idea is that the inference algorithm is not given direct access to the program, but performs \emph{queries} on it. In the Hoare-query model, algorithms repeatedly choose $\alpha,\beta$ and query for the validity of Hoare triples $\hoaretrip{\alpha}{\tr}{\beta}$, where $\tr$ is the transition relation denoting one step of the system, inaccessible to the algorithm but via such Hoare queries. The check itself is implemented by an oracle, which in practice is a SAT solver.
This model is general enough to capture algorithms such as PDR and its variants, and leaves room for other interesting design choices, but does not capture white-box approaches such as abstract interpretation~\cite{DBLP:conf/popl/CousotC77}.
The advantage of this model for a theoretical study is that it enables an information-based analysis,
which
\begin{inparaenum}[(i)]
	\item sidesteps open computational complexity questions, and therefore results in unconditional lower bounds on the computational complexity of SAT-based algorithms captured by the model, and
	\item grants meaning to questions about generalization from partial information we discuss later.
\end{inparaenum}

\paragraph{\bf Results}
This research addresses two main questions related to the core ideas behind PDR, and theoretically analyzes them in the context of the Hoare-query model:
\begin{enumerate}
	\item \label{q:gen} These algorithms revolve around the question of \emph{generalization}: from observing concrete states (to be excluded from the invariant), the algorithm seeks to produce assertions that hold for \emph{all} reachable states.
	The different heuristics in this context are largely understood as clever ways of performing this generalization.
	The situation is similar in interpolation-based algorithms, only that generalization is performed from bounded safety proofs rather than states.
	How should generalization be performed to achieve efficient invariant inference?

	\item \label{q:ind} %
	A key aspect of PDR is the form of SAT checks it uses, as part of \emph{relative inductiveness checks}, of Hoare triples $\hoaretrip{\alpha}{\tr}{\beta}$ in which in general $\alpha \neq \beta$.\footnote{For the PDR-savvy: $\beta$ is typically a candidate clause, and $\alpha$ is derived from the previous frame.} %
	Repeated queries of this form are potentially richer than presenting a series of candidate invariants, where the check is $\hoaretrip{\alpha}{\tr}{\alpha}$.
	Is there a benefit in using relative inductiveness beyond inductiveness checks?
\end{enumerate}
We analyze these questions in the foundational case of Boolean programs, \begin{changebar}which is applicable to infinite-state systems through predicate abstraction~\cite{DBLP:conf/cav/GrafS97,DBLP:conf/popl/FlanaganQ02,DBLP:conf/cade/LahiriQ09},\end{changebar} and is also a core part of other invariant inference techniques for infinite-state systems~\cite[e.g.][]{DBLP:conf/sat/HoderB12,DBLP:conf/cav/KomuravelliGC14,DBLP:journals/jacm/KarbyshevBIRS17}.

In \Cref{sec:hoare-lower-bound}, we answer question \ref{q:gen} with an impossibility result, by showing that no choice of generalization can lead to an inference algorithm using only a polynomial number of Hoare queries.
Our lower bound is information-theoretic, and holds even with unlimited computational power, showing that the problem of generalization is chiefly a question of information gathering.

In \Cref{sec:rice-vs-ice}, we answer question \ref{q:ind} in the affirmative, by showing an
exponential gap between algorithms utilizing rich $\hoaretrip{\alpha}{\delta}{\beta}$ checks and algorithms that
perform only inductiveness checks $\hoaretrip{\alpha}{\delta}{\alpha}$. Namely, we
construct a class of programs for which a simple version of PDR can infer invariants efficiently,
but every algorithm learning solely from counterexamples to the inductiveness of candidates requires an
exponential number of queries.
This result shows, for the first time theoretically, the significance of relative inductiveness checks %
as the foundation of PDR's mechanisms,
in comparison to a machine learning approach pioneered in the ICE model~\cite{ICELearning,DBLP:conf/popl/0001NMR16} that infers invariants based on inductiveness checks only (but of course this result does not mean that PDR is \emph{always} more efficient than every ICE algorithm).

Our results also clarify the relationship between the problem of invariant inference and the classical theory of exact concept learning with queries~\cite{DBLP:journals/ml/Angluin87}.
\begin{changebar}
In particular, our results imply that learning from counterexamples to induction is harder than learning from positive \& negative examples (\Cref{sec:concept-vs-invariant}), providing a formal justification to the existing intuition~\cite{ICELearning}.
This demonstrates that the convergence rate of learning in Counterexample-Guided Inductive Synthesis~\cite[e.g.][]{DBLP:conf/asplos/Solar-LezamaTBSS06,DBLP:conf/icse/JhaGST10,DBLP:journals/acta/JhaS17} depends on the form of examples.
We also establish impossibility results for directly applying algorithms from concept learning to invariant inference.
\end{changebar}

The contributions of the paper are summarized as follows:
\begin{itemize}
	\item We define the problem of polynomial-length invariant inference, and show it is $\psigma{2}$-complete (\Cref{sec:inference-decision-problem}), strengthening the hardness result of template-based abstraction by \citet{DBLP:conf/cade/LahiriQ09}.

\item %
We introduce  the Hoare-query model, a black-box model of invariant inference capable of modeling PDR (\Cref{sec:hoare-query-algorithms}), and study the query complexity of polynomial-length invariant inference in this model.

	\item We show that in general an algorithm in this model requires an exponential number of queries to solve polynomial-length inference, even though Hoare queries are rich and versatile (\Cref{sec:hoare-lower-bound}). %

	\item We also extend this result to a model capturing interpolation-based algorithms (\Cref{sec:interpolation-query-algorithms}).

	\item We show that Hoare queries are more powerful than inductiveness queries (\Cref{sec:rice-vs-ice}). This also proves that ICE learning cannot model PDR, and that the extension of the model by~\citet{DBLP:conf/vmcai/VizelGSM17} is necessary.

	\item %
	\begin{changebar}
	We prove that exact learning from counterexamples to induction is harder than exact learning from positive \& negative examples, and derive impossibility results for translating some exact concept learning algorithms to the setting of invariant inference (\Cref{sec:concept-vs-invariant}).
	\end{changebar}
\end{itemize}

\OMIT{

\sharon{general comment: the paper deals with two orthogonal points: first, it addresses a variant of invariant inference, where the size of the invariant matters. Second, it consider oracles, i.e., specific ways of accessing the transition relation to model SAT-based algorithms. Since these two things are orthogonal, how can we justify combining them instead of focusing on one or on each of them independently?}
\yotam{Perhaps through the motivation: understanding the complexity of invariant inference algorithms, which (1) actually construct an invariant, not only prove its existence, and (2) do that in a SAT-based way. ?}

\yotam{I later brag about lower bounds holding even when candidates are exponentially long. we are looking into computationally unrestricted procedures, so we maybe we should say that it is important to consider invariant of length $n$ because otherwise the search space is larger than exponential, and then it is trivial that it takes exponential time to narrow it down (e.g.\ finding a Boolean function amongst all $2^{2^n}$)}

Invariant inference is a search problem: given a transition system, find an inductive invariant proving it safe, if one exists.
To be able to prove interesting programs, modern inference algorithms employ large classes of candidate invariants, seeking expressive power and generality. However, a broad class of candidates mandates effective search strategies:
somewhere in the haystack of candidates a correct invariant must be found.
Many complex invariant inference algorithms have been developed to tackle this challenge, including interpolation~\cite{DBLP:conf/cav/McMillan03}, IC3/PDR~\cite{ic3,pdr}, and abduction~\cite{DBLP:conf/oopsla/DilligDLM13}.
\TODO{this is already specific: SAT-based. We need to narrow down the scope to that}

The area of effective invariant inference is still developing. \TODO{explain difficulties}

What can we say theoretically about these algorithms?

\subsubsection{The Complexity of Invariant Inference}
We want to study the complexity of inferring an invariant, in terms of the size of the program and/or
the size of the smallest invariant. There are reasons.
$\psigma{2}$ hardness.
Say we don't care about the innate complexity of a SAT call.
\neil{I think what you want to say here is that even though NP complete problems can be intractible,
  there has been substantial practical progress using SAT and SMT solvers, so proving that a problem
  is NP complete will not be considered a proof that it is intractible, nor that it is easy, but
  rather an indication that we should cautiously proceed.}

\paragraph{\bf Related work.}
The complexity of invariant inference has been studied by \citet{DBLP:conf/cade/LahiriQ09}. They
show that deciding the existence of an invariant, possibly of exponential size, is PSPACE-complete,
and that the problem of template-based inference is $\psigma{2}$-complete. Polynomial-length inference for CNF formulas can be encoded as specific instances of template-based inference; the $\psigma{2}$-hardness proof of \citet{DBLP:conf/cade/LahiriQ09} uses more general templates and therefore does not directly imply the $\psigma{2}$-hardness of polynomial length inference. They also show that inference is only $\ppi{1}=\coNP$-complete when
candidates are only conjunctions, (or, dually, disjunctions).
\neil{When they are disjunctions is it still NP complete or now co-NP complete?} \yotam{they are actually both coNP, thanks!}
In this paper we focus on inferring
invariants of length $n$ from richer syntactical classes.

\citet{DBLP:conf/popl/PadonISKS16} analyze classes of programs and candidate invariants where the safety problem is undecidable, and consider the decidability of finding an invariant in specific syntactic classes in these cases. In this paper we also take into account the length of the target invariant, leading to the study of complexity rather than of decidability.

\subsubsection{Hoare-Query Invariant Inference Algorithms}
Blackbox access, SAT-based inference, not abstract interpretation (in standard form).
}
\section{Overview}
\OMIT{	
Backward reachability. How to generalize? We can try simple scheme, that of PDR-1, it doesn't always work because it treats every state non-reachable in one step as completely unreachable. That's why PDR is more complicated.
But must generalize, otherwise it's exponential.

Example system: parity, error is all 1's, invariant that lsb=0, can find this quickly if you know. if you don't, might choose very strange generalizations.

Our results are not specific to backward reachability, this is just an illustration here.
}
Coming up with inductive invariants is one of the most challenging tasks of formal verification---it
is often referred to as the \emph{``Eureka!''} step.
This paper studies the asymptotic complexity of automatically inferring CNF invariants of polynomial length, a problem we call \textbf{\textit{polynomial-length inductive invariant inference}}, in a SAT-based black-box model.

Consider the dilemmas Abby faces when she attempts to develop an algorithm %
for this problem from first principles.
Abby is excited about the popularity of SAT-based inference algorithms.
Many such algorithms operate by repeatedly performing checks of Hoare triples of the form $\hoaretrip{\alpha}{\tr}{\beta}$, where $\alpha,\beta$ are a precondition and postcondition (resp.) chosen by the algorithm in each query and $\tr$ is the given transition relation (loop body). A SAT solver implements the check. %
We call such checks \textbf{\textit{Hoare queries}}, and focus in this paper on \emph{black-box} inference algorithms in the \textbf{\textit{Hoare-query model}}: algorithms that access the transition relation solely through Hoare queries.

\iflong

\begin{figure*}[t]
  \centering
\begin{minipage}{\textwidth}
  \begin{minipage}{\textwidth}
  \begin{lstlisting}[numbers=left, numberstyle=\tiny, numbersep=5pt, escapeinside={(*}{*)}, xleftmargin=3.0ex]
init $x_1 = \ldots = x_n = 0$
axiom $\exists! i, \, 1\leq i \leq n. \ c_i = 1$

function add-double($\mathbf{a}$,$\mathbf{b}$) = $(\mathbf{a} + 2 \cdot \mathbf{b}) \, \bmod{2^n}$

while *
   input $y_1,\ldots,y_n$
   if $c_1$:
      $(x_1,x_2,\ldots,x_{n-1},x_n)$ $\ \ \,$ := add-double($(x_1,x_2,\ldots,x_{n-1},x_n)$,$(y_1,y_2,\ldots,y_{n-1},y_n)$)
   if $c_2$:
      $(x_2,x_3,\ldots,x_n,x_1)$     $\ \ \ \ \;$ := add-double($(x_2,x_3,\ldots,x_n,x_1)$,$(y_2,y_3,\ldots,y_n,y_1)$)
   ...
   if $c_n$:
      $(x_n,x_1,\ldots,x_{n-2},x_{n-1})$ $\,$ := add-double($(x_n,x_1,\ldots,x_{n-2},x_{n-1})$,$(y_n,y_1,\ldots,y_{n-2},y_{n-1})$)
   assert $\neg (x_1 = \ldots = x_n = 1)$
  \end{lstlisting}
\end{minipage}
\captionof{figure}{\footnotesize An example propositional transition system for which we would like to infer an inductive invariant. The state is over $x_1,\ldots,x_n$. The variables $y_1,\ldots,y_n$ are inputs and can change arbitrarily in each step. $c_1,\ldots,c_n$ are immutable, with the assumption that exactly one is true.}
  \label{fig:add-even}
  \end{minipage}
\end{figure*}  \fi
\Cref{fig:add-even} displays one example program that Abby is interested in inferring an inductive invariant for.
In this program, a number $\mathbf{x}$, represented by $n$ bits, is initialized to zero, and at each iteration incremented by an even number that is decided by the input variables $\mathbf{y}$ (all computations are mod $2^n$). The representation of the number $\mathbf{x}$ using the bits $x_1,\ldots,x_n$ is determined by another set of bits $c_1,\ldots,c_n$, which are all immutable, and only one of them is true: if $c_1=\true$, the number is represented by $x_1,x_2,\ldots,x_n$, if $c_2=\true$ the least-significant bit (lsb) shifts and the representation is $x_2,x_3,\ldots,x_n,x_1$ and so on.
The safety property is that $\mathbf{x}$ is never equal to the number with all bits 1. Intuitively, this holds because the number $\mathbf{x}$ is always even. An \emph{inductive invariant} states this fact, taking into account the differing representations, by stating that the lsb (as chosen by $\mathbf{c}$) is always 0:
$
	I = (c_1 \rightarrow \neg x_1) \land \ldots (c_n \rightarrow \neg x_n).
$
Of course, Abby aims to verify many systems, of which \Cref{fig:add-even} is but one example.

\subsection{Example: Backward-Reachability with Generalization}
\label{sec:overview-gen-dilemma}
\iflong
\begin{figure*}[t]
\begin{footnotesize}
\vspace*{-\baselineskip}
\begin{minipage}[t]{0.5\textwidth}
\begin{algorithm}[H]
\caption{\\Backward-reachability}
\label{alg:backward-reach}
\begin{algorithmic}[1]
\Procedure{Block-Cube}{$\tr$}
	\State $I \gets \neg\Bad$
	\While{$\hoaretrip{I}{\tr}{I}$ not valid}
		\State $\sigma,\sigma' \gets \Call{cti}{\tr,I}$ %
		\State $\label{algln:min-ret}$ $d \gets$ \Call{Block}{$\tr$,$\sigma$}
		\State $\label{algln:backgen-strengthen}$ $I \gets I \land \neg d$
	\EndWhile
	\Return $I$
\EndProcedure
\end{algorithmic}%
\end{algorithm}%
\vspace{-0.7cm}%
\begin{algorithm}[H]
\caption{Naive Block}
\label{alg:block-backward-reach}
\begin{algorithmic}[1]
\Procedure{Block-Cube}{$\tr$,$\sigma$}
	\State \Return $\bigwedge_{i, \, \sigma \models p_i}{p_i} \land 
			\bigwedge_{i, \, \sigma \models \neg p_i}{\neg p_i}$
\EndProcedure
\end{algorithmic}
\end{algorithm}%
\end{minipage}%
\begin{minipage}[t]{0.5\textwidth}
\begin{algorithm}[H]
\caption{\\Generalization with Init-Step Reachability}
\label{alg:pdr1-gen}
\begin{algorithmic}[1]
\Procedure{Block-PDR-1}{$\tr,\sigma$}
	\State $\label{algln:genpdr1:cube}$ $d \gets \Call{cube}{\sigma}$
	\For{$l \in \Call{cube}{\sigma}$}
		\State $t \gets d \setminus \set{l}$
		$\label{algln:genpdr1:gen-cond}$ \If{$(\Init \implies \neg t) \land \hoaretrip{\Init}{\tr}{\neg t}$}
			\State $d \gets t$
		\EndIf
	\EndFor
	\Return $d$
\EndProcedure
\end{algorithmic}
\end{algorithm}
\end{minipage}%
\end{footnotesize}
\end{figure*} \fi
How should Abby's algorithm go about finding inductive invariants?
One known strategy is that of \emph{backward reachability}, in which the invariant is strengthened to exclude states from which bad states may be reachable.\footnote{Our results are not specific to backward-reachability algorithms; we use them here for motivation and illustration.}
\Cref{alg:backward-reach} is an algorithmic backward-reachability scheme: it repeatedly checks for the existence of a counterexample to induction (a transition $\sigma,\sigma'$ of $\tr$ from $\sigma \models I$ to $\sigma' \not\models I$), and strengthens the invariant to exclude the pre-state $\sigma$ using the formula \textsc{Block} returns.

\Cref{alg:backward-reach} depends on the choice of \textsc{Block}.
The most basic approach is of \Cref{alg:block-backward-reach}, which excludes \emph{exactly} the pre-state, by conjoining to the invariant the negation of the cube of $\sigma$ (the cube is the conjunction of all literals that hold in the state; the only state that satisfies $\cube{\sigma}$ is $\sigma$ itself\sharon{can remove to save space:}, and thus the only one to be excluded from $I$ in this approach).
For example, when \Cref{alg:backward-reach} needs to block the state $\mathbf{x} = 011\ldots1, \mathbf{c}=000\ldots1$ (this state reaches the bad state $\mathbf{x} = 111\ldots1, \mathbf{c}=000\ldots1$), \Cref{alg:block-backward-reach} does so by conjoining to the invariant %
the negation of $\neg x_n \land x_{n-1} \land x_{n-2} \land \ldots x_1 \land \neg c_n \land \neg c_{n-1} \land \neg c_{n-2} \land \ldots c_1$, \sharon{can remove to save space:}and this is a formula that other states do not satisfy.

Alas, \Cref{alg:backward-reach} with blocking by \Cref{alg:block-backward-reach} is not efficient. In essence it operates by enumerating and excluding the states backward-reachable from bad. The number of such states is potentially exponential, making \Cref{alg:block-backward-reach} unsatisfactory. For instance, the example of \Cref{fig:add-even} requires the exclusion of all states in which $\mathbf{x}$ is odd for every choice of lsb, a number of states exponential in $n$. The algorithm would thus require an exponential number of queries to arrive at a (CNF) inductive invariant, even though a CNF invariant with only $n$ clauses exists (as above).

Efficient inference hence requires Abby to exclude more than a single state at each time, namely, to \emph{generalize} from a counterexample---as real algorithms do.
What generalization strategy could Abby choose that would lead to efficient invariant inference?

\subsection{All Generalizations Are Wrong}
\iflong
\else

\begin{figure*}[t]
  \centering
\begin{minipage}{\textwidth}
  \begin{minipage}{\textwidth}
  \begin{lstlisting}[numbers=left, numberstyle=\tiny, numbersep=5pt, escapeinside={(*}{*)}, xleftmargin=3.0ex]
init $x_1 = \ldots = x_n = 0$
axiom $\exists! i, \, 1\leq i \leq n. \ c_i = 1$

function add-double($\mathbf{a}$,$\mathbf{b}$) = $(\mathbf{a} + 2 \cdot \mathbf{b}) \, \bmod{2^n}$

while *
   input $y_1,\ldots,y_n$
   if $c_1$:
      $(x_1,x_2,\ldots,x_{n-1},x_n)$ $\ \ \,$ := add-double($(x_1,x_2,\ldots,x_{n-1},x_n)$,$(y_1,y_2,\ldots,y_{n-1},y_n)$)
   if $c_2$:
      $(x_2,x_3,\ldots,x_n,x_1)$     $\ \ \ \ \;$ := add-double($(x_2,x_3,\ldots,x_n,x_1)$,$(y_2,y_3,\ldots,y_n,y_1)$)
   ...
   if $c_n$:
      $(x_n,x_1,\ldots,x_{n-2},x_{n-1})$ $\,$ := add-double($(x_n,x_1,\ldots,x_{n-2},x_{n-1})$,$(y_n,y_1,\ldots,y_{n-2},y_{n-1})$)
   assert $\neg (x_1 = \ldots = x_n = 1)$
  \end{lstlisting}
\end{minipage}
\captionof{figure}{\footnotesize An example propositional transition system for which we would like to infer an inductive invariant. The state is over $x_1,\ldots,x_n$. The variables $y_1,\ldots,y_n$ are inputs and can change arbitrarily in each step. $c_1,\ldots,c_n$ are immutable, with the assumption that exactly one is true.}
  \label{fig:add-even}
  \end{minipage}
\end{figure*}  %
\begin{figure*}[t]
\begin{footnotesize}
\vspace*{-\baselineskip}
\begin{minipage}[t]{0.5\textwidth}
\begin{algorithm}[H]
\caption{\\Backward-reachability}
\label{alg:backward-reach}
\begin{algorithmic}[1]
\Procedure{Block-Cube}{$\tr$}
	\State $I \gets \neg\Bad$
	\While{$\hoaretrip{I}{\tr}{I}$ not valid}
		\State $\sigma,\sigma' \gets \Call{cti}{\tr,I}$ %
		\State $\label{algln:min-ret}$ $d \gets$ \Call{Block}{$\tr$,$\sigma$}
		\State $\label{algln:backgen-strengthen}$ $I \gets I \land \neg d$
	\EndWhile
	\Return $I$
\EndProcedure
\end{algorithmic}%
\end{algorithm}%
\vspace{-0.7cm}%
\begin{algorithm}[H]
\caption{Naive Block}
\label{alg:block-backward-reach}
\begin{algorithmic}[1]
\Procedure{Block-Cube}{$\tr$,$\sigma$}
	\State \Return $\bigwedge_{i, \, \sigma \models p_i}{p_i} \land 
			\bigwedge_{i, \, \sigma \models \neg p_i}{\neg p_i}$
\EndProcedure
\end{algorithmic}
\end{algorithm}%
\end{minipage}%
\begin{minipage}[t]{0.5\textwidth}
\begin{algorithm}[H]
\caption{\\Generalization with Init-Step Reachability}
\label{alg:pdr1-gen}
\begin{algorithmic}[1]
\Procedure{Block-PDR-1}{$\tr,\sigma$}
	\State $\label{algln:genpdr1:cube}$ $d \gets \Call{cube}{\sigma}$
	\For{$l \in \Call{cube}{\sigma}$}
		\State $t \gets d \setminus \set{l}$
		$\label{algln:genpdr1:gen-cond}$ \If{$(\Init \implies \neg t) \land \hoaretrip{\Init}{\tr}{\neg t}$}
			\State $d \gets t$
		\EndIf
	\EndFor
	\Return $d$
\EndProcedure
\end{algorithmic}
\end{algorithm}
\end{minipage}%
\end{footnotesize}
\end{figure*} \fi
\label{sec:generalization-partial-info-motivation}
One simple generalization strategy Abby considers appears in \Cref{alg:pdr1-gen}, based on the standard ideas in IC3/PDR~\cite{ic3,pdr} \sharon{can remove for space:}and subsequent developments~\cite[e.g.][]{DBLP:conf/sat/HoderB12,DBLP:conf/cav/KomuravelliGC14}.
It starts with the cube (as \Cref{alg:block-backward-reach}) and attempts to drop literals, resulting in a smaller conjunction, which many states satisfy; all these states are excluded  from the candidate in \cref{algln:backgen-strengthen} of \Cref{alg:backward-reach}. Hence with this generalization \Cref{alg:backward-reach} can exclude many states in each iteration\sharon{can remove for space:}, overcoming the problem with the naive algorithm above. %
\Cref{alg:pdr1-gen} chooses to drop a literal from the conjunction if no state reachable in at most one step from $\Init$ satisfies the conjunction even when that literal is omitted (\cref{algln:genpdr1:gen-cond} of \Cref{alg:pdr1-gen}); %
we refer to this algorithm as PDR-1, since it resembles PDR with a single frame.

For example, when in the example of \Cref{fig:add-even} the algorithm attempts to block the state with $\mathbf{x} = 011\ldots1, \mathbf{c}=000\ldots1$, \Cref{alg:pdr1-gen} minimizes the cube to $d = x_1 \land c_1$, because no state reachable in at most one step satisfies $d$, but this is no longer true when another literal is omitted.
Conjoining the invariant with $\neg d$ (in \cref{algln:backgen-strengthen} of \Cref{alg:backward-reach}) produces a clause of the invariant, $c_1 \rightarrow \neg x_1$.
In fact, our results show that PDR-1 finds the aforementioned invariant in $n^2$ queries.
\OMIT{
In fact, we show that in cases such as \Cref{fig:add-even}, where (1) the set of reachable states can be expressed as a monotone formula with $n$ clauses, and (2) every reachable state is reachable in at most one step (in this example, the reachable states are the even numbers and they are reachable in one step), then PDR-1 is guaranteed to find an inductive invariant in $O(n^2)$ queries (\Cref{lem:pdr1-complexity}).\footnote{
	An inductive invariant has the advantage over Bounded Model Checking (BMC) even for systems with diameter one, because once they are established they constitute a sound proof of safety independent of the assumption that the diameter is one. The result reads that PDR-1 always finds a sound safety proof for systems satisfying (1),(2), although it might fail for other systems.
}
}

Yet there is a risk in over-generalization, that is, of dropping \emph{too many} literals and excluding too many states. In \Cref{alg:backward-reach}, generalization must not return a formula that some reachable states satisfy, or the candidate $I$ would exclude reachable states and would not be an inductive invariant.
\Cref{alg:pdr1-gen} chooses to take the strongest conjunction that does not exclude any state reachable in at most one step; it is of course possible (and plausible) that some states are reachable in two steps but not in one. \Cref{alg:backward-reach} with the generalization in \Cref{alg:pdr1-gen} might fail in such cases.

The necessity of generalization, on the one hand, and the problem of over-generalization on the other leads in practice to complex heuristic techniques. Instead of simple backward-reachability with generalization per \Cref{alg:backward-reach}, PDR never commits to a particular generalization~\cite{pdr} through a sequence of \emph{frames}, which are (in some sense) a sequence of candidate invariants. The clauses resulting from generalization are used to strengthen frames according to a bounded reachability analysis; \Cref{alg:pdr1-gen} corresponds to generalization in the first frame. %

Overall, the study of backward-reachability and the PDR-1 generalization leaves us with the question:
\textbf{\textit{Is there a choice of generalization that can be used---in any way---to achieve an efficient invariant inference algorithm?}}

In a non-interesting way, the answer is \emph{yes}, there is a ``good'' way to generalize:
Use \Cref{alg:backward-reach}, with the following generalization strategy: Upon blocking a pre-state $\sigma$, compute an inductive invariant of polynomial length, and return the clause of the invariant that excludes $\sigma$,\footnote{Such a clause exists because $\sigma$ is backward-reachable from bad states, and thus excluded from the invariant.}
and this terminates in a polynomial number of steps. %

Such generalization is clearly unattainable. It requires (1) perfect information of the transition system, and (2) solving a computationally hard problem, since we show that polynomial-length inference is $\psigma{2}$-hard (\Cref{thm:sigma2-hardness}).
What happens when generalization is computationally unbounded (an arbitrary function), but operates based on \emph{partial information} of the transition system?
Is there a \textbf{\textit{generalization from partial information}}, be it computationally intractable, that facilitates \textbf{\textit{efficient inference}}?
If such a generalization exists we may wish to view invariant inference heuristics as \emph{approximating} it in a computationally efficient way. %

Similar questions arise in interpolation-based algorithms, only that generalization is performed not from a concrete state, but from a bounded unreachability proof. Still it is challenging to generalize enough to make progress but not too much as to exclude reachable states (or include states from which bad is reachable).

\OMIT{
\emph{Ideal generalization}: yes, if you know everything and can compute everything. We obviously don't.

\emph{Computational limitation}: no, unless $\psigma{2} = \psigma{1}$.
}

\subsubsection{Our Results}
\label{sec:overview-results-hoare-lower}
\begin{changebar}
Our first main result in this paper
\end{changebar}is that \textbf{\textit{in general, there does not exist a generalization scheme from partial information}} leading to efficient inference based on Hoare queries.
Technically, we prove that even a computationally unrestricted generalization from information gathered from Hoare queries requires an exponential number of queries.
This result applies to any generalization strategy and any algorithm using it that can be modeled using Hoare queries, including \Cref{alg:backward-reach} as well as more complex algorithms such as PDR.
We also extend this lower bound to a model capturing interpolation-based algorithms (\Cref{thm:interpolation-nopoly}).

These results are surprising because a-priori it would seem possible, using unrestricted computational power, to devise queries that repeatedly halve the search space, yielding an invariant with a polynomial number of queries (the number of candidates is only exponential because we are interested in invariants up to polynomial length). %
We show that this is impossible to achieve using Hoare queries. %

\subsection{Inference Using Rich Queries}
\label{overview:rice-vs-ice}
\OMIT{
Even what we used for the previous example generalization is beyond inductiveness. Is it important? Oh yes. (Ask questions also earlier.)
}

So far we have established strong impossibility results for invariant inference based on Hoare queries in the general case, even with computationally unrestricted generalization.
We now turn to shed some light on the techniques that inference algorithms such as PDR employ in practice.
One of the fundamental principles of PDR is the incremental construction of invariants relying on rich Hoare queries.
PDR-1 demonstrates a simplified realization of this principle. When PDR-1 considers a clause to strengthen the invariant, it checks the reachability of that individual clause from  $\Init$, rather than the invariant as a whole. This is the Hoare query $\hoaretrip{\Init}{\tr}{\neg t}$ in \cref{algln:genpdr1:gen-cond} of \Cref{alg:pdr1-gen}, in which, crucially, the precondition is different from the postcondition. The full-fledged PDR is similar in this regard, strengthening a frame according to reachability from the previous frame via relative induction checks~\cite{ic3}.

The algorithm in \Cref{alg:block-backward-reach} is fundamentally different, and uses only inductiveness queries $\hoaretrip{I}{\tr}{I}$, a specific form of Hoare queries where the precondition and postcondition are the same.
Algorithms performing only inductiveness checks can in fact be very sophisticated, traversing the domain of candidates in clever ways.
This approach was formulated in the \emph{ICE learning} framework for learning inductive invariants~\cite{ICELearning,DBLP:conf/popl/0001NMR16} \sharon{remove for space:}(later extended to general Constrained-Horn Clauses~\cite{DBLP:journals/pacmpl/EzudheenND0M18}), in which algorithms present new candidates based on positive, negative, and implication examples returned by a ``teacher'' in response to incorrect candidate invariants.\footnote{Our formulation focuses on implication examples---counterexamples to inductiveness queries---and strengthens the algorithm with full information about the set of initial and bad states instead of positive and negative examples (resp.).} %
The main point is that such algorithms do not perform queries other than inductiveness, and choose the next candidate invariant based solely on the counterexamples to induction showing the previous candidates were unsuitable.

The contrast between the two approaches raises the question: \textbf{\textit{Is there a benefit to invariant inference in Hoare queries richer than inductiveness?}}
For instance, to model PDR in the ICE framework, \citet{DBLP:conf/vmcai/VizelGSM17} extended the framework with relative inductiveness checks, but the question whether such an extension is necessary remained open.

\subsubsection{Our Results}
\label{sec:overview-results-inductiveness}
\begin{changebar}
Our second significant result in this paper
\end{changebar}is showing an exponential gap between the general Hoare-query model and the more specific inductiveness-query model. To this end, we construct a class of transition systems, including the example of \Cref{fig:add-even}, for which
\begin{inparaenum}
	\item PDR-1, which is a Hoare-query algorithm, infers an invariant in a polynomial number of queries, but
	\item \emph{every} inductiveness-query algorithm requires an \emph{exponential} number of queries, that is, an exponential number of candidates before it finds a correct inductive invariant.
\end{inparaenum}
This demonstrates that analyzing the reachability of clauses separately can offer an exponential advantage in certain cases.
This also proves that PDR cannot be cast in the ICE framework, and that the extension by~\citet{DBLP:conf/vmcai/VizelGSM17} is necessary and strictly increases the power of inference with a polynomial number of queries. %
To the best of our knowledge, this is not only the first lower bound on ICE learning demonstrating such an exponential gap (also see the discussion in \Cref{sec:related-work}), but also the first polynomial upper bound on PDR for a class of systems. %
We show this separation on a class of systems constructed using a technical notion of \emph{maximal systems for monotone invariants}. These are systems for which there exists a monotone invariant (namely, an invariant propositional variables appear only negatively) with a linear number of clauses, and the transition relation includes \emph{all} transitions allowed by this invariant. For example, a maximal system can easily be constructed from~\Cref{fig:add-even}: this system allows every transition between states satisfying the invariant (namely, between all even $\mathbf{x}$'s with the same representation), and also every transition between states violating the invariant (namely, between all odd $\mathbf{x}$'s with the same representation).\footnote{Transitions violating the $\mathbf{c}$ axiom or modifying it are excluded in this modeling.}; a maximal system also includes all the transitions from states that violate that invariant to the states that satisfy it (here, between odd $\mathbf{x}$ and even $\mathbf{x}$ with the same $\mathbf{c}$). %
The success of PDR-1 on such systems relies on the small diameter (every reachable state is reachable in one step) and harnesses properties of prime consequences of monotone formulas. %
In contrast, we show that for inductiveness-query algorithms this class is as hard as the class of \emph{all} programs admitting monotone invariants, whose hardness is established from the results of \Cref{sec:overview-results-hoare-lower}.
For example, from the perspective of inductiveness-query algorithms, the example of \Cref{fig:add-even}, which is a maximal program as explained above, %
is as hard as \emph{any} system that admits its invariant (and also respects the $\mathbf{c}$ axiom and leaves $\mathbf{c}$ unchanged). %
This is because an inductiveness-query algorithm can only benefit from having \emph{fewer} transitions and hence fewer counterexamples to induction, whereas maximal programs include as many transitions as possible.
If an inductiveness query algorithm is to infer an invariant for the example of \Cref{fig:add-even}, it must also be able to infer an invariant for all systems whose transitions are a \emph{subset} of the transitions of this example. This includes systems with an exponential diameter, %
as well as systems admitting other invariants, potentially exponentially long.
This program illustrates our lower bound construction, which takes \emph{all} maximal programs for monotone-CNF invariants.

In our lower bound we follow the existing literature on the analysis of inductiveness-query
algorithms, which focuses on the worst-case notion w.r.t.\ potential examples (strong convergence
in~\citet{ICELearning}). An interesting direction is to analyze inductiveness-query algorithms that
exercise some control over the choice of counterexamples to induction, or under probabilistic
assumptions on the distribution of examples.

\subsection{A Different Perspective: Exact Learning of Invariants with Hoare Queries}
This paper can be viewed as developing a theory of exact learning of inductive invariants with Hoare queries, akin to the classical theory of concept learning with queries~\cite{DBLP:journals/ml/Angluin87}.
The results outlined above are consequences of natural questions about this model:
The impossibility of generalization from partial information (\Cref{sec:overview-results-hoare-lower}) stems from an exponential lower bound on the Hoare-query model. %
The power of rich Hoare queries (\Cref{sec:overview-results-inductiveness}) is demonstrated by an exponential separation between the Hoare- and inductiveness-query models, in the spirit of the gap between concept learning using both equivalence and membership queries and concept learning using equivalence queries alone~\cite{DBLP:journals/ml/Angluin90}. %

The similarity between invariant inference (and synthesis in general) and exact concept learning has been observed before~\cite[e.g.][]{DBLP:conf/icse/JhaGST10,ICELearning,DBLP:journals/acta/JhaS17,DBLP:series/natosec/AlurBDF0JKMMRSSSSTU15,DBLP:conf/colt/BshoutyDVY17}. %
Our work highlights some interesting differences and connections between invariant learning with Hoare, and concept learning with equivalence and membership queries.
\begin{changebar}
This comparison yields (im)possibility results for translating algorithms from concept learning with queries to invariant inference with queries.
Another outcome is the third significant result of this paper: a proof that learning
from counterexamples to induction is inherently harder than learning from examples labeled
as positive or negative, formally corroborating the intuition advocated by~\citet{ICELearning}.
More broadly, the complexity difference between learning from labeled examples and learning from counterexamples to induction demonstrates that the convergence rate of learning
in Counterexample-Guided Inductive Synthesis~\cite[e.g.][]{DBLP:conf/asplos/Solar-LezamaTBSS06,DBLP:conf/icse/JhaGST10,DBLP:journals/acta/JhaS17} depends on the form of examples.
The proof of this result builds on the lower bounds discussed earlier, and is discussed in \Cref{sec:concept-vs-invariant}.\end{changebar}\footnote{
	\begin{changebar}
	It may also be interesting to note that one potential difference between classical learning and invariant inference, mentioned by~\citet{DBLP:conf/tacas/LodingMN16}, does not seem to manifest in the results discussed in \Cref{sec:overview-results-hoare-lower}: the transition systems in the lower bound for inductiveness queries in \Cref{cor:monmax-lower} have a \emph{unique} inductive invariant, and still the problem is hard.
	\end{changebar}
}

\OMIT{
Comparing learning invariants from a class of formulas to learning
We show that learning invariants in the class of \sharon{polynomial} monotone CNF invariants %
requires an \emph{exponential} number of Hoare queries, even when exponentially long queries are allowed (\Cref{thm:query-nopoly}), whereas in concept learning
\begin{enumerate}
	\item %
Using membership queries and equivalence queries the class of monotone CNF formulas can be learned in a \emph{polynomial} number of queries~\cite{DBLP:journals/ml/Angluin87}.

	\item With just equivalence queries, \sharon{the even larger class of} CNF formulas can be learned in \emph{subexponential} number of queries~\cite{DBLP:journals/ipl/Bshouty96}.

	\item %
    Allowing exponentially-long equivalence queries, a \emph{polynomial} number of equivalence queries suffice to learn CNF formulas through the halving algorithm/majority vote~\cite{barzdins1972prediction,DBLP:journals/ml/Angluin87,DBLP:journals/ml/Littlestone87}.
\end{enumerate}
These differences reflect the fact that in invariant learning the queries access the transition relation rather than candidate invariant(s) directly. Transition relations with complex reachability patterns (e.g.\ an exponential diameter) incur challenges beyond identifying the invariant as a formula. Further, a counterexample to induction can be handled by either strengthening or weakening the invariant, so learning from them is more challenging than from positive/negative examples.
\sharon{ replace with much shorter and focused: Another important difference is the fact that a counterexample to induction can be handled in two ways---strengthen the invariant to exclude the pre-state, or weaken the invariant to include the post-state---and there is no telling which is the correct choice. This is unlike the situation with equivalence queries, where the learner knows whether the example the teacher returns is negative or positive. In maximal systems the algorithm can perform a membership query by checking reachability from the initial states in at most one step, but this is not possible in general transition systems (as \citet{ICELearning} also point out). The feasibility of membership queries in maximal systems lies at the heart of our upper bound on inference for maximal systems for monotone invariants, and PDR-1 can be viewed as an incarnation of the polynomial algorithm for concept learning monotone formulas~\cite{DBLP:journals/ml/Angluin87}. } \yotam{wrote}

\sharon{should this be said in the section on concept learning: Another important difference is the fact that a counterexample to induction can be handled in two ways---strengthen the invariant to exclude the pre-state, or weaken the invariant to include the post-state---and there is no telling which is the correct choice. This is unlike the situation with equivalence queries, where the learner knows whether the example the teacher returns is negative or positive.}
}

\section{Background}
\subsection{States, Transitions Systems, and Inductive Invariants}
In this paper we consider safety problems defined via formulas in propositional logic. Given a propositional vocabulary $\voc$ that consists of a finite set of Boolean variables, we denote by $\Formulas{\voc}$ the set of well formed propositional formulas defined over $\voc$. A \emph{state} is a valuation to $\voc$. For a state $\sigma$, the \emph{cube of $\sigma$}, denoted $\cube{\sigma}$, is the conjunction of all literals that hold in $\sigma$.
A \emph{transition system} is a triple $\TS = (\Init,\tr,\Bad)$ such that $\Init,\Bad \in \Formulas{\voc}$ define the \emph{initial states} and the \emph{bad states}, respectively, and $\tr \in \Formulas{\voc \uplus \voc'}$ defines the \emph{transition relation}, where $\voc' = \{ \prop' \mid \prop \in \voc\}$ is a copy of the vocabulary used to describe the post-state of a transition.
A class of transition systems, denoted $\progs{}$, is a set of transition systems.
A transition system $\TS$ is \emph{safe} if all the states that are reachable from the initial states via steps of $\tr$ satisfy $\neg \Bad$. %
An \emph{inductive invariant} for $\TS$ is a formula $I \in \Formulas{\voc}$ such that $\Init \implies I$, $I \wedge \tr \implies I'$, and $I \implies \neg\Bad$, where $I'$ %
denotes the result of substituting each $\prop \in \voc$ for $\prop' \in \voc'$ in $I$, and $\varphi \implies \psi$ denotes the validity of the formula $\varphi \to \psi$. In the context of propositional logic, a transition system is safe if and only if it has an inductive invariant. %
When $I$ is not inductive, a \emph{counterexample to induction} is a pair of states $\sigma, \sigma'$ such that $\sigma,\sigma'\models I \wedge \tr \wedge \neg I'$ (where the valuation to $\voc'$ is taken from $\sigma'$).

\paragraph{The classes $\cnf{n}{}$, $\dnf{n}{}$ and $\moncnf{n}{}$}
$\cnf{n}{}$ is the set of propositional formulas in Conjunctive Normal Form (CNF) with at most $n$ clauses (disjunction of literals). $\dnf{n}{}$ is likewise for Disjunctive Normal Form (DNF), where $n$ is the maximal number of cubes (conjunctions of literals).
$\moncnf{n}{}$ is the subset of $\cnf{n}{}$ in which all literals are negative. %

\subsection{Invariant Inference Algorithms}
In this section we briefly provide background on inference algorithms that motivate our theoretical development in this paper. The main results of the paper do not depend on familiarity with these algorithms or their details; this (necessarily incomprehensive) ``inference landscape'' is presented here for context and motivation
for defining the Hoare-query model (\Cref{sec:hoare-query-algorithms}), studying its complexity and the feasibility of generalization (\Cref{sec:hoare-lower-bound}), and analyzing the power of Hoare queries compared to inductiveness queries (\Cref{sec:rice-vs-ice}). We allude to specific algorithms in motivating each of these sections.

\subsubsection*{\bf {IC3/PDR}}
\label{sec:prelim-pdr}

IC3/PDR maintains a sequence of formulas $F_0,F_1,\ldots$, called \emph{frames}, each of which can be understood as a candidate inductive invariant.
The sequence is gradually modified and extended throughout the algorithm's run. It is maintained as an approximate reachability sequence, meaning that
	(1) $\Init \implies F_0$,
	(2) $F_j \implies F_{j+1}$,
	(3) $F_j \land \tr \implies (F_{j+1})'$, and
    (4) $F_j \implies \neg \Bad$.
These properties ensure that $F_j$ overapproximates the set of states reachable in $j$ steps, and that the approximations contain no bad states.
(We emphasize that $F_j \implies \neg \Bad$ does \emph{not} imply that a bad state is unreachable in \emph{any} number of states.)
The algorithm terminates when one of the frames implies its preceding frame ($F_j \implies F_{j-1}$), in which case it constitutes an inductive invariant, or when a counterexample trace is found.
In iteration $N$, a new frame $F_N$ is added to the sequence. One way of doing so is by initializing $F_N$ to $\true$, and strengthening it until it excludes all bad states.
Strengthening is done by \emph{blocking} bad states: given a bad state $\sigma_b \models F_N \land \Bad$, the algorithm strengthens $F_{N-1}$ to exclude all $\sigma_b$'s pre-states---states that satisfy $F_{N-1} \land \tr \land (\cube{\sigma_b})'$---one by one (thereby demonstrating that $\sigma_b$ is unreachable in $N$ steps). Blocking a pre-state $\sigma_a$ from frame $N-1$ is performed by a recursive call to block its own pre-states from frame $N-2$, and so on. If this process reaches a state from $\Init$, the sequence of states from the recursive calls constitutes a trace reaching $\Bad$ from $\Init$, which is a counterexample to safety.
Alternatively, when a state $\sigma$ is successfully found to be unreachable from $F_{j-1}$ in one step, i.e., $F_{j-1} \land \tr \land (\cube{\sigma_b})'$ is unsatisfiable, frame $F_j$ is strengthened to reflect this fact.
Aiming for efficient convergence (see \Cref{sec:overview-gen-dilemma}), PDR chooses to \emph{generalize}, and exclude more states. A basic form of generalization is performed by dropping literals from $\cube{\sigma}$ as long as the result $t$ is still unreachable from $F_{j-1}$, i.e., $F_{j-1} \land \tr \land t'$ is still unsatisfiable. This is very similar to PDR-1 above (\Cref{sec:generalization-partial-info-motivation}), where $F_{j-1}$ was always $F_0 = \Init$. Often \emph{inductive generalization} is used, dropping literals as long as %
$F_{j-1} \land \neg t \land \tr \land t'$, reading that $\neg t$ is inductive \emph{relative} to $F_{j-1}$, which can drop more literals than basic generalization.
A core optimization of PDR is \emph{pushing}, in which a frame $F_j$ is ``opportunistically'' strengthened with a clause $\alpha$ from $F_{j-1}$, if $F_{j-1}$ is already sufficiently strong to show that $\alpha$ is unreachable in $F_j$.

For a more complete presentation of PDR and its variants as a set of abstract rules that may be applied nondeterministically see e.g.~\citet{DBLP:conf/sat/HoderB12,DBLP:conf/fmcad/GurfinkelI15}.
The key point from the perspective of this paper is that the algorithm and its variants access the transition relation $\tr$ in a very specific way, checking whether some $\alpha$ is unreachable in one step of $\tr$ from the set of states satisfying a formula $F$ (or those satisfying $F \land \alpha$), and obtains a counterexample when it is reachable
(see also~\citet{DBLP:conf/vmcai/VizelGSM17}).
Crucially, other operations (e.g., maintaining the frames, checking whether $F_j \implies F_{j-1}$, etc.) do not use $\tr$. %
We will return to this point when discussing the Hoare-query model, which can capture IC3/PDR (\Cref{sec:hoare-query-algorithms}).

\subsubsection*{\bf ICE}
\label{sec:prelim-ice}
The ICE framework~\cite{ICELearning,DBLP:conf/popl/0001NMR16} (later extended to general Constrained-Horn Clauses~\cite{DBLP:journals/pacmpl/EzudheenND0M18}), is a learning framework for inferring invariants from positive, negative and implication counterexamples.
We now review the framework using the original terminology and notation; later in the paper we will use a related formulation that emphasizes the choice of candidates (in \Cref{sec:inductiveness-query-algorithms}). %

In ICE learning, the teacher holds an unknown target $(P,N,R)$, where $P,N \subseteq D, R \subseteq D \times D$ are sets of examples. The learner's goal is to find a hypothesis $H \in C$ s.t.\ $P \subseteq H, N \cap H = \emptyset$, and for each $(x,y) \in R$, $x \in H \implies y \in H$. The natural way to cast inference in this framework %
is, given a transition system $(\Init,\tr,\Bad)$ and a set of candidate invariants $\mathcal{L}$, to take $D$ as the set of program states, $P$ a set of reachable states including $\Init$, $N$ a set of states including $\Bad$ from which a safety violation is reachable, $R$ the set of transitions of $\tr$, and $C = \mathcal{L}$.
Iterative ICE learning operates in rounds. In each round, the learner is provided with a sample---$(E,B,I)$ s.t.\ $E \subseteq P, B \subseteq N, I \subseteq R$---and outputs an hypothesis $H \in C$. The teacher returns that the hypothesis is correct, or extends the sample with an example showing that $H$ is incorrect.
The importance of implication counterexamples is that they allow implementing a teacher using a SAT/SMT solver without ``guessing'' what a counterexample to induction indicates~\cite{ICELearning,DBLP:conf/tacas/LodingMN16}.
Examples of ICE learning algorithms include Houdini~\cite{DBLP:conf/fm/FlanaganL01} and symbolic abstraction~\cite{DBLP:conf/vmcai/RepsSY04,DBLP:journals/entcs/ThakurLLR15}, as well as designated algorithms~\cite{ICELearning,DBLP:conf/popl/0001NMR16}.
Theoretically, the analysis of \citet{ICELearning} focuses on \emph{strong convergence} of the learner, namely, that the learner can always reach a correct concept, no matter how the teacher chooses to extend samples between rounds.
In this work, we will be interested in the \emph{number of rounds} the learner performs. We will say that the learner is \emph{strongly-convergent} with \emph{round-complexity} $r$ if for every ICE teacher, the learner finds a correct hypothesis in at most $r$ rounds, provided that one exists. We extend this definition to a class of target descriptions in the natural way.

\subsubsection*{\bf {Interpolation}}
\label{sec:prelim-interpolation}
The idea of interpolation-based algorithms, first introduced by~\citet{DBLP:conf/cav/McMillan03}, is to generalize proofs of bounded unreachability into elements of a proof of \emph{un}bounded reachability, utilizing \emph{Craig interpolation}. Briefly, this works as follows: encode a bounded reachability from a set of states $F$ in $k$ steps, and use a SAT solver to find that this cannot reach $\Bad$. When efficient interpolation is supported in the logic and solver, the SAT solver can produce an interpolant $C$: a formula representing a set of states that
\begin{inparaenum}[(i)]
	\item overapproximates the set of states reachable from $F$ in $k_1$ steps, and still
	\item cannot reach $\Bad$ in $k_2$ steps
\end{inparaenum}
(any choice $k_1+k_2=k$ is possible).
Thus $C$ overapproximates concrete reachability from $F$ without reaching a bad state, although both these facts are known in only a bounded number of steps. The hope is that $C$ would be a useful generalization to include as part of the invariant. The original algorithm~\cite{DBLP:conf/cav/McMillan03} sets some $k$ as the current unrolling bound, starts with $F=\Init$, obtains an interpolant $C$ with $k_1=1,k_2=k-1$, sets $F \gets F \lor C$ and continues in this fashion, until an inductive invariant is found, or $\Bad$ becomes reachable in $k$ steps from $F$, in which case $k$ is incremented and the algorithm is restarted. The use of interpolation and generalization from bounded unreachability has been used in many works since~\cite[e.g.][]{DBLP:conf/fmcad/VizelG09,DBLP:conf/cav/McMillan06,DBLP:journals/lmcs/JhalaM07,DBLP:conf/popl/HenzingerJMM04,DBLP:conf/tacas/VizelGS13}. Combining ideas from interpolation and PDR has also been studied~\cite[e.g.][]{DBLP:conf/cav/VizelG14}. The important point for this paper is that many interpolation-based algorithms only access the transition relation when checking bounded reachability (from some set of states $\alpha$ to some set of states $\beta$), and extracting interpolants when the result is unreachable. We will return to this point when discussing the interpolation-query model, which aims to capture interpolation-based algorithms (\Cref{sec:interpolation-query-algorithms}).   
\section{Polynomial-Length Invariant Inference}
\label{sec:inference-decision-problem}
In this section we formally define the problem of \emph{polynomial-length invariant inference} for CNF formulas, which is the focus of this paper. We then relate the problem to the problem of inferring DNF formulas with polynomially many cubes via duality (see \refappendix{sec:forward-backward-duality}), and focus on the case of CNF in the rest of the paper.

\OMIT{

\paragraph{Parameterized classes of candidate invariants.}
To enable restricting the length of the inductive invariants being inferred,
we parameterize the language of candidate invariants by $n \in \Nat$, a parameter that determines the (upper bound on the) target length of the invariant.
Formally, rather than considering a single set $\Linv{}$ of candidate invariants, we consider sequences $\seq{\Linv{n}}{n \in \Nat}$ where for every $n$, $\Linv{n}$ %
is a set of formulas of length polynomial in $n$ and the vocabulary. Namely, for every $n\in \mathbb{N}$ and vocabulary $\voc$, all the formulas in $\Linv{n}$ over $\voc$
are represented by at most $p(n, \card{\voc})$ bits in some fixed representation, for some fixed polynomial $p(\cdot,\cdot)$. %
In this paper we are interested in classes $\seq{\Linv{n}}{n \in \Nat}$ where $\Linv{n}$ contains all CNF formulas with at most $n$ clauses:
}

Our object of study is the problem of polynomial-length inference:
\begin{definition}[Polynomial-Length Inductive Invariant Inference]
\label{def:bounded-inference}
The \emph{polynomial-length inductive invariant inference problem} (\emph{invariant inference} for short) for a class of transition systems  $\progs{}$ and a polynomial %
$p(n) = \Omega(n)$
is the problem:
Given a transition system $\TS \in \progs{}$ over $\voc$, \emph{decide} whether there exists an inductive invariant
$I \in \cnf{p(n)}{}$
for $\TS$, where $n = \card{\voc}$.
\end{definition}

\paragraph{Notation.} In the sequel, when considering the polynomial-length inductive invariant inference problem of a transition system $\TS = (\Init,\tr,\Bad) \in \progs{}$, we denote by $\voc$ the vocabulary of $\Init,\Bad$ and $\tr$. Further, we denote $n = \card{\voc}$.

\subsubsection*{\bf Complexity} The \emph{complexity} of polynomial-length inference is measured in %
$\card{\TS} = \card{\Init} + \card{\tr} + \card{\Bad}$. %
Note that the invariants are required to be polynomial in $n=\card{\voc}$.
$\cnf{p(n)}{}$ is a rich class of invariants. Inference in more restricted classes can be solved efficiently.
For example, when only conjunctive candidate invariants are considered, %
and $\progs{}$ is the set of all propositional transition systems, the problem can be decided in a polynomial number of SAT queries through the Houdini algorithm~\cite{DBLP:conf/fm/FlanaganL01,DBLP:conf/cade/LahiriQ09}.
Similar results hold also for CNF formulas with a constant number of literals per clause (by defining a new predicate for each of the polynomially-many possible clauses and applying Houdini), and for CNF formulas with a constant number of clauses (by translating them to DNF formulas with a constant number of literals per cube and applying the dual procedure).
However, a restricted class of invariants may miss invariants for some programs and reduces the generality of the verification procedure.
Hence in this paper we are interested in the richer class of polynomially-long CNF invariants. In this case the problem is no longer tractable even with a SAT solver:
\begin{theorem}
\label{thm:computational-hardness}
\label{thm:sigma2-hardness}
Let $\progs{}$ be the set of all propositional transition systems.
Then polynomial-length inference for $\progs{}$ %
is $\psigma{2}$-complete,
where  $\psigma{2} = \text{NP}^{\text{\,SAT}}$ is the second level of the polynomial-time hierarchy.
\end{theorem}
We defer the proof to \Cref{sec:proof-of-comp-hardness}.
We note that polynomial-length inference can be encoded as specific instances of template-based inference; the $\psigma{2}$-hardness proof of \citet{DBLP:conf/cade/LahiriQ09} uses more general templates and therefore does not directly imply the $\psigma{2}$-hardness of polynomial-length inference.
Lower bounds on polynomial-length inference entail lower bounds for template-based inference.

\begin{remark}
In the above formulation, an efficient procedure for deciding \emph{safety} does not imply polynomial-length inference is tractable, since the program may be safe, but all inductive invariants may be too long. %
To overcome this technical quirk, we can consider a \emph{promise problem}~\cite{DBLP:conf/birthday/Goldreich06a} variant of polynomial-length inference:

Given a transition system $\TS \in \tsclass$,
\begin{itemize}
	\item \textbf{(Completeness)} If $\TS$ has an inductive invariant $I \in \cnf{p(n)}{}$, the algorithm must return \emph{yes}. %
	\item \textbf{(Soundness)} If $\TS$ is not safe the algorithm must return \emph{no}.
\end{itemize}
Other cases, including the case of safety with an invariant outside $\cnf{p(n)}{}$, are not constrained.
An algorithm deciding safety thus solves also this problem.
All the results of this paper apply both to the standard version above and the promise problem: upper bounds on the standard version trivially imply upper bounds on the promise problem, and in our lower bounds we use transition systems that are either (i) safe \emph{and} have an invariant in $\cnf{p(n)}{}$, or (ii) unsafe.
\end{remark}

\OMIT{
\begin{remark}[Decision and search]
\label{lem:decide-by-search}
\TODO{classic results about decision vs. search, because it's NP-complete?}
To further motivate the formulation of polynomial-length inference as a decision problem, we note that an algorithm for invariant \emph{search} induces an algorithm for invariant \emph{decision}.
Thus a lower bound on the decision problem induces a lower bound on invariant \emph{search}.
This holds also when considering for search only systems that admit an invariant in $\cnf{p(n)}{}$, namely, even when omitting all systems for which the decision is ``no''.
This is because given an algorithm for invariant search for $\progs{}$ %
we can decide $\progs{} \uplus \badtsclass{}$, %
where $\badtsclass{}$ is the set of unsafe transition systems \yotam{make sure it doesn't collide}, as follows: run the search algorithm, and check the resulting invariant. If the result is either error or not actually an inductive invariant, then return ``no'', otherwise ``yes''.
When we prove upper bounds on invariant decision in \Cref{sec:monotone-upper} we use such an algorithm.
\end{remark}
}

\section{Invariant Inference with Queries and the Hoare Query Model}
\label{sec:hoare-query-algorithms}

In this paper we study algorithms for polynomial-length inference through \emph{black-box} models of \emph{inference with queries}.
In this setting, the algorithm accesses the transition relation through (rich) queries, but cannot read the transition relation directly. Our main model is of \emph{Hoare-query algorithms}, which query the validity of a postcondition from a precondition in one step of the system. Hoare-query algorithms faithfully capture a large class of SAT-based invariant inference algorithms, including PDR and related methods.
A black-box model of inference algorithms facilitates an analysis of the \emph{information} of the transition relation the algorithm acquires. %
The advantage is that such an information-based analysis sidesteps open computational complexity questions, and therefore results in unconditional lower bounds on the computational complexity of SAT-based algorithms captured by the model. Such an information-based analysis is also necessary for questions involving unbounded computational power and restricted information, in the context of computationally-unrestricted bounded-reachability generalization (see \Cref{sec:gen-intro}).

In this section we define the basic notions of queries and query-based inference algorithms. We also define the primary query model we study in the paper: the Hoare-query model.
In the subsequent sections we introduce and study additional query models: the interpolation-query model (\Cref{sec:interpolation-query-algorithms}), and the inductiveness-query model (\Cref{sec:inductiveness-query-algorithms}).

\subsubsection*{\bf Inference with queries} %
We model queries of the transition relation in the following way:
A \emph{query oracle} $Q$ is an oracle that accepts a transition relation $\tr$, as well as additional inputs, and returns some output.
The additional inputs and the output, together also called the \emph{interface} of the oracle, depend on the query oracle under consideration. A \emph{family} of query oracles $\qoraclefam$ is a set of query oracles with the same interface.
We consider several different query oracles, representing different ways of obtaining information about the transition relation.

\begin{definition}[Inference algorithm in the query model] %
\label{def:inference-query-algorithm}
An \emph{inference algorithm from queries}, denoted $\bbalg{Q}{}(\Init,\Bad,\bb{\tr})$, is defined w.r.t.\ a query oracle $Q$ and is given:
\begin{itemize}
	\item access to the query oracle $Q$,
	\item the set of initial states ($\Init$) and bad states ($\Bad$);
	\item \begin{changebar}the transition relation $\tr$, encapsulated---hence the notation $\bb{\tr}$---meaning that the algorithm cannot access $\tr$ (not even read it) except for extracting its vocabulary; $\tr$ can only be passed as an argument to the query oracle $Q$.\end{changebar}
\end{itemize}
$\bbalg{Q}{}(\Init,\Bad,\bb{\tr})$ solves the problem of polynomial-length invariant inference for $(\Init,\tr,\Bad)$.
\end{definition}

\subsubsection*{\bf The Hoare-query model}

Our main object of study in this paper is the Hoare-query model of invariant inference algorithms. It captures SAT-based invariant inference algorithms querying the behavior of a single step of the transition relation at a time. %
\begin{definition}[Hoare-Query Model] \label{def:hoareQuery}
For a transition relation $\tr$ and input formulas $\alpha, \beta\in \Formulas{\Sigma}$, the
\emph{Hoare-query  oracle}, $\horacle{\tr}{\alpha,\beta}$,  returns $\false$ if $(\alpha \land \tr
\land \lnot \beta') \in \sat$; otherwise it returns $\true$.

An algorithm in the \emph{Hoare-query model}, also called a \emph{Hoare-query algorithm}, is an inference from queries algorithm
expecting the Hoare query oracle. %
\end{definition}
Intuitively, a Hoare-query algorithm gains access to the transition relation, $\tr$, exclusively by
repeatedly choosing $\alpha, \beta\in \Formulas{\Sigma}$,  and calling
$\horacle{\tr}{\alpha,\beta}$.

If we are using a $\sat$ solver to compute the Hoare-query, $\horacle{\tr}{\alpha,\beta}$, then when
the answer is $\false$, the $\sat$ solver will also produce a counterexample pair of states
$\sigma,\sigma'$ such that $\sigma,\sigma' \models \alpha \land \tr \land \neg \beta'$.

We observe that using binary search, a Hoare-query algorithm can do the same:

\begin{lemma}
\label{lem:hoare-cti}
Whenever $\horacle{\tr}{\alpha,\beta} = \false$, a Hoare-query algorithm can find $\sigma,\sigma'$ such that $\sigma,\sigma' \models \alpha \land \tr \land \neg \beta'$ using $n=\card{\voc}$ Hoare queries.
\end{lemma}
\begin{proof}
For each  $\prop_i \in \voc \uplus \voc'$, if $\prop_i \in \voc$, conjoin it to $\alpha$, else to
$\beta$, and check whether $\horacle{\tr}{\alpha_i,\beta_i}$ is still $\false$. If it is, continue
to $\prop_{i+1}$; otherwise flip $\prop_i$ and continue to $\prop_{i+1}$.
\end{proof}

\begin{changebar}
\paragraph{Example: PDR as a Hoare-query algorithm}
The Hoare-query model captures the prominent PDR algorithm, facilitating its theoretical analysis.
As discussed in \Cref{sec:prelim-pdr}, PDR accesses the transition relation via checks of unreachability in one step and counterexamples to those checks. These operations are captured in the Hoare query model by checking %
$\horacle{\tr}{F,\alpha}$ or $\horacle{\tr}{F \land \alpha,\alpha}$ (for the algorithm's choice of $F, \alpha \in \Formulas{\voc}$), %
and obtaining a counterexample using a polynomial number of Hoare queries, if one exists (\Cref{lem:hoare-cti}).
Furthemore, the Hoare-query model is general enough to express a broad range of PDR variants %
that differ in the way they use such checks %
but still access the transition relation only through such queries.

\medskip
The Hoare-query model is not specific to PDR. It also captures algorithms in the ICE learning model~\cite{ICELearning}, as we discuss in \Cref{sec:inductiveness-query-algorithms}, and as result can model algorithms captured by the ICE model (see~\Cref{sec:prelim-ice}). %
In \Cref{sec:ind-vs-hoare-all} we show that the Hoare-query model is in fact strictly more powerful than the ICE model.
\end{changebar}

\begin{remark}
Previous black-box models for invariant inference~\cite{ICELearning} encapsulated access also to $\Init,\Bad$. In our model we encapsulate only access to $\tr$, since (1)~it is technically simpler, (2)~a simple transformation can make $\Init,\Bad$ uniform across all programs, embedding the differences in the transition relation; indeed, our constructions of classes of transition systems in this paper are such that $\Init,\Bad$ are the same in all transition systems that share a vocabulary, hence $\Init,\Bad$ may be inferred from the vocabulary.
(Unrestricted access to $\Init,\Bad$ is stronger, thus lower bounds on our models apply also to models restricting access.)
\end{remark}

\subsubsection*{\bf Complexity.}
Focusing on \emph{information}, we do not impose computational restrictions on the algorithms, and only count the number of queries the algorithm performs to reveal information of the transition relation. In particular, when establishing lower bounds on the query complexity, we even consider algorithms that may compute non-computable functions. However, whenever we construct algorithms demonstrating upper bounds on query complexity, these algorithms in fact have polynomial \emph{time} complexity, and we note this when relevant.

Given a query oracle and an inference algorithm that uses it, we analyze the number of queries the algorithm performs as a function of $n= \card{\voc}$, in a \emph{worst-case} model w.r.t.\ to possible transition systems over $\voc$ in the class of interest.

The definition is slightly more complicated by considering, as we do later in the paper, query-models in which more than one oracle exists, i.e., an algorithm may use any oracle from a \emph{family} of query oracles. In this case, we analyze the query complexity of an algorithm in a \emph{worst-case} model w.r.t.\ the possible query oracles in the family as well.

Formally, the query complexity is defined as follows:
\begin{definition}[Query Complexity]
\label{def:query-complexity}
For a class of transitions systems $\tsclass$, the \emph{query complexity} of (a possibly computationally unrestricted) $\bbalg{}{}$ w.r.t.\ a query oracle family $\qoraclefam$ is defined as
\begin{equation}
\querycomplexity_{\bbalg{}{}}^{\qoraclefam}(n) = \sup_{Q \in \qoraclefam}{\sup_{\substack{(\Init,\tr,\Bad) \in \tsclass,\\\card{\voc} = n}}{\mbox{\#query}(\bbalg{Q}{}(\Init,\Bad,\bb{\tr}))}}
\end{equation}
where $\mbox{\#query}(\bbalg{Q}{}(\Init,\Bad,\bb{\tr}))$ is the number of times the algorithm accesses $Q$ given this oracle and the input. (These numbers might be infinite.)
\end{definition}

The query complexity in the Hoare-query model is %
$\querycomplexity_{\bbalg{}{}}^{\set{\horaclesym}}(n)$.

\begin{remark}
\label{rem:short-tr}
In our definition, query complexity is a function of the size of the vocabulary $n=\card{\voc}$, but not of the size of the representation of the transition relation $\card{\tr}$. This reflects the fact that an algorithm in the black-box model does not access $\tr$ directly. %
In \refappendix{sec:poly-tr} we discuss the complexity w.r.t.\ $\card{\tr}$ as well. The drawback of such a complexity measure is that learning $\tr$ itself becomes feasible, undermining the black-box model. Efficiently learning $\tr$ is possible when using unlimited computational power and exponentially-long queries. However, whether the same holds when using unlimited computational power with only polynomially-long queries is related to open problems in classical concept learning.
\end{remark}

\section{The Information Complexity of Hoare-Query Algorithms}
\label{sec:hoare-lower-bound}
In this section we prove an information-based lower bound on Hoare-query invariant inference algorithms, and also extend the results to algorithms using \emph{interpolation}, another SAT-based operation. %
We then apply these results to study the role of information in generalization as part of inference algorithms.

\subsection{Information Lower Bound for Hoare-Query Inference}
\label{sec:information-bounds}

We show that a Hoare-query inference algorithm requires $\explowerbound$ Hoare-queries in the worst case to decide whether a CNF invariant of length polynomial in $n$  exists.
(Recall that $n$ is a shorthand for $\card{\voc}$, the size of the vocabulary of the input transition system.) %
This result applies even when allowing the choice of queries to be inefficient, and when allowing the queries to use exponentially-long formulas.
It provides a lower bound on the time complexity of actual algorithms, such as PDR, that are captured by the model.
Formally:

\begin{theorem}
\label{thm:hoare-nopoly}
Every Hoare-query inference algorithm $\bbalg{\horaclesym}{}$ deciding polynomial-length inference for
the class of all propositional transition systems %
has query complexity of $\explowerbound$.
\end{theorem}

The rest of this section proves a strengthening of this theorem, for a specific class of transition systems (which we construct next), for any class of invariants that includes monotone CNF,
and for computationally-unrestricted algorithms:

\begin{theorem}
\label{thm:query-nopoly}
\label{thm:query-nopoly-mon}
Every Hoare-query inference algorithm $\bbalg{\horaclesym}{}$, even computationally-unrestricted, deciding invariant inference for the class of transition systems $\hardclass$ (\Cref{sec:hard-class}) and for any class of target invariants $\mathcal{L}$ s.t.\ $\moncnf{n}{} \subseteq \mathcal{L}$\footnote{%
Here we extend the definition of polynomial-length invariant inference to $\mathcal{L}$ instead of $\cnf{p(n)}{}$.}
has query complexity of $\explowerbound$.
\end{theorem}
\begin{changebar}
(That classes containing $\moncnf{n}{}$ are already hard becomes important in \Cref{sec:rice-vs-ice}.)
\end{changebar}

\subsubsection{A Hard Class of Transition Systems}
\label{sec:hard-class}
In this section we construct a $\hardclass$, a hard class of transition systems, on which we prove hardness results.

\paragraph{The $\qbf{2}$ problem} The construction of $\hardclass$ follows the $\psigma{2}$-complete problem of $\qbf{2}$  from classical computational complexity theory. In this problem, the input is a quantified Boolean formula $\exists \vec{y}. \, \forall \vec{x}. \ \phi(\vec{x},\vec{y})$ where $\phi$ is a Boolean (quantifier-free) formula, and the problem of $\qbf{2}$ is to decide whether the quantified formula is \emph{true}, namely, there exists a Boolean assignment to $\vec{y}$ s.t.\ $\phi(\vec{x},\vec{y})$ is true for every Boolean assignment to $\vec{x}$.

\paragraph{The class $\hardclass$.}
For each $k \in \mathbb{N}$, we define $\hardclass^k$. Finally, $\hardclass = \bigcup_{k \in \mathbb{N}} \hardclass^k$.

Let $k \in \mathbb{N}$. %
For each formula $\exists \vec{y}. \, \forall \vec{x}. \, \phi(\vec{y},\vec{x})$, where $\vec{y} = y_1,\ldots,y_k$, $\vec{x} = x_1,\ldots,x_k$ are variables and $\phi$ is a quantifier-free formula over the variables $\vec{x} \cup \vec{y}$, we define a transition system $\fts{\phi} = (\Init_k, \psigmatr{\phi},\Bad_k)$. Intuitively, it iterates through $\vec{y}$ lexicographically, and for each $\vec{y}$ it iterates lexicographically through $\vec{x}$ and checks if all assignments to $\vec{x}$ satisfy $\phi(\vec{y},\vec{x})$. If no such $\vec{y}$ is found, this is an error. More formally,
\begin{enumerate}
	\item $\voc_k = \{y_1,\ldots,y_k,x_1,\ldots,x_k,a,b,e\}$.
	\item $\Init_k = \vec{y} = \vec{0} \land \vec{x} = \vec{0} \land \neg a \land b \land \neg e$.
	\item $\Bad_k = e$.
	\item $\psigmatr{\phi}$: evaluate $\phi(\vec{y},\vec{x})$, and perform the following changes (at a single step):
	If the result is false, set $a$ to true.
	If $\vec{x} = \vec{1}$ and $a$ is still false, set $b$ to false.
	If in the pre-state $\vec{x} = \vec{1}$, increment $\vec{y}$ lexicographically, reset $a$ to false, and set $\vec{x} = 0$; otherwise increment $\vec{x}$ lexicographically.
	If in the pre-state $\vec{y} = \vec{1}$, set $e$ to $b$.
	(Intuitively, $a$ is false as long as no falsifying assignment to $\vec{x}$ has been encountered for the current $\vec{y}$, $b$ is true as long as we have not yet encountered a $\vec{y}$ for which there is no falsifying assignment.)
\end{enumerate}
We denote the resulting class of transition systems $\ftset^k = \{\fts{\phi} \mid \phi = \phi(y_1,\ldots,y_k,x_1,\ldots,x_k)\}$.

The following lemma relates the $\qbf{2}$ problem for $\phi$ to the inference problem of $\fts{\phi}$:
\begin{lemma}
\label{lem:validity-vs-inference}
Let $\fts{\phi} \in \ftset^k$. Then $\fts{\phi}$ is safe iff it has an inductive invariant in $\moncnf{2k+1}{}$ iff the formula $\exists \vec{y}. \, \forall \vec{x}. \, \phi(\vec{y},\vec{x})$ is true.
\end{lemma}
\begin{proof}
There are two cases:
\begin{itemize}
	\item If $\exists \vec{y}. \, \forall \vec{x}. \, \phi(\vec{y},\vec{x})$ is true, let $\vec{v}$ be the first valuation for $\vec{y}$ that realizes the existential quantifiers. Then the following is an inductive invariant for $\fts{\phi}$:
	\begin{equation}
	\label{eq:sigma2-inv}
		I = \neg e
			\land (b \rightarrow \vec{y} \leq \vec{v})
			\land ((b \land a) \rightarrow \vec{y} < \vec{v})
	\end{equation}
	where the lexicographic constraint is expressed by the following recursive definition on $\vec{y}_{[d]} = (y_1,\ldots,y_d),\vec{v}_{[d]} = (v_1,\ldots,v_d)$:
	\begin{equation*}
		\vec{y}_{[d]} < \vec{v}_{[d]} \ \eqdef \
		\begin{cases}
			\neg y_d \lor (\vec{y}_{[d-1]} < \vec{v}_{[d-1]})	& v_d = \true
			\\
			\neg y_d \land (\vec{y}_{[d-1]} < \vec{v}_{[d-1]}) & v_d = \false
		\end{cases}
	\end{equation*}
	and $\vec{y} \leq \vec{v} \overset{\triangle}{=} \vec{y} < (\vec{v} + 1)$ (or $\true$ if $\vec{v} = \vec{1}$).

    $I \in \moncnf{2k+1}{}$: Note that $\vec{y}_{[k]} < \vec{v}_{[k]}$ can be written in CNF with at most $n$ clauses: in the first case a literal is added to each clause, and in the second another clause is added. Thus $I$ can be written in CNF with at most $2k+1$ clauses. Further, the literals of $\vec{y}$ appear only negatively in $\vec{y}_{[k]} < \vec{v}_{[k]}$, and hence also in $I$. The other literals ($\neg e, \neg a, \neg b$) also appear only negatively in $I$. Hence, $I$ is monotone.

	\iflong
	$I$ is indeed an inductive invariant: initiation and safety are straightforward. For consecution, consider a valuation to $\vec{y}$ in a pre-state satisfying the invariant. (We abuse notation and refer to the valuation by $\vec{y}$.)
	There are three cases:
	\begin{itemize}
		\item
			If $\vec{y} < \vec{v}$, then
			\begin{inparaenum}[(i)]
				\item $\neg e$ is retained by a step,
				\item $\vec{y} \leq \vec{v}$ holds after a step,
				\item $\vec{y} < \vec{v}$ still holds unless the transition is from the last evaluation for $\vec{y} = \vec{v} - 1$ to $\vec{v}$, in which case $a$ is turned to $\false$.
			\end{inparaenum}

		\item If $\vec{y} > \vec{v}$, the invariant guarantees that in the pre-state $b$ is false, and thus $e$ remains false after a step. $b$ also remains false and thus the rest of the invariant also holds in the post-state.

		\item If $\vec{y} = \vec{v}$, the invariant guarantees that in the pre-state either $b$ is false or $a$ is false. If $b$ is false the same reasoning of the previous case applies. Otherwise, we have that $a$ is false. By the definition of $\vec{v}$ all valuations for $\vec{x}$ results in $\phi(\vec{v},\vec{x})=\true$, so $a$ remains false after a step, and once we finish iterating through $\vec{x}$ we set $b$ to false immediately.
	\end{itemize}
	The claim follows.
	\else
	The proof that $I$ is indeed inductive appears in the extended version~\cite{extendedVersion}.
	\fi

	\item If $\exists \vec{y}. \, \forall \vec{x}. \, \phi(\vec{y},\vec{x})$ is not true, then $\fts{\phi}$ is not safe (and thus does not have an inductive invariant of any length). %
This is because for every valuation of $\vec{y}$ a violating $\vec{x}$ is found, turning $a$ to true, and $b$ never turns to false, so after iterating through all possible $\vec{y}$'s $e$ will become true.
\end{itemize}
\end{proof}

Before we turn to prove \Cref{thm:query-nopoly} and establish a lower bound on the query complexity in the Hoare model, we note that this construction also yields the computational hardness mentioned in \Cref{sec:inference-decision-problem}:
\begin{proof}[Proof of \Cref{thm:sigma2-hardness}]
\label{sec:proof-of-comp-hardness}
The upper bound is straightforward: guess an invariant in $\cnf{p(n)}{}$ %
and check it.
For the lower bound, use the reduction outlined above: given $\phi(y_1,\ldots,y_k,x_1,\ldots,x_k)$, construct $\fts{\phi}$. Note that the vocabulary size, $n$, is $2k+3$, and the invariant, when exists, is of length at most $2k+1 \leq n$.\footnote{For an arbitrary polynomial $p(n) =\Omega(n)$, e.g., $p(n) = c \cdot n$ with $0< c <1$, enlarge $\voc$, e.g., by adding to $\Init$ initialization of fresh variables that are not used elsewhere, to ensure existence of an invariant of length $\leq p(n)$.} The reduction is polynomial as the construction of $\fts{\phi}$ (and $n$) from $\phi$ is polynomial in $k$ and $\card{\phi}$: note that lexicographic incrementation can be performed with a propositional formula of polynomial size. %
\end{proof}
 
\subsubsection{Lower Bound's Proof}

We now turn to prove \Cref{thm:query-nopoly}.
\begin{changebar}
Given an algorithm with polynomial query complexity, the proof constructs two transition system: one that has a polynomial-length invariant and one that does not, and yet all the queries the algorithm performs do not distinguish between them. The construction uses the path the algorithm takes when all Hoare queries return $\false$ %
as much as possible. Intuitively, such responses are less informative and rule out less transition relations, because they merely indicate the existence of a single counterexample to a Hoare triple, opposed to the result $\true$ which indicates that \emph{all} transitions satisfy a property.
\end{changebar}
\OMIT{
	As a warmup (with interest of its own), we prove the following information-based lower bound on solving $\psigma{2}$ instances through $\psigma{1}$ queries:
	\begin{theorem}
	\label{thm:s2-s1-lower-bound}
	Deciding $\qbf{2}$ of $\exists \vec{y}. \, \forall \vec{x}. \, \phi(\vec{y},\vec{x})$ where $\card{\vec{x}} = \card{\vec{y}} = n$ requires $\explowerbound$ queries in a black-box model performing queries of the form $\textit{SAT}_{?[\phi]}(\theta)$ (the algorithm chooses $\theta$, and it can be exponentially long), each answered by whether the formula $\theta(\vec{y},\vec{x}) \land \phi(\vec{y},\vec{x})$ is satisfiable or not.
	\end{theorem}
	\begin{proof}
	Consider a set of queries $\textit{SAT}_{?[\phi]}(\theta_1),\ldots,\textit{SAT}_{?[\phi]}(\theta_m)$, $m < 2^n$.
	We construct two instances, $\Theta_1 = \exists \vec{y}. \, \forall \vec{x}. \, \phi_1(\vec{y},\vec{x})$ and $\Theta_2 = \exists \vec{y}. \, \forall \vec{x}. \, \phi_2(\vec{y},\vec{x})$, such that:
	\begin{enumerate}[(i)]
		\item $\Theta_1$ is true.
		\item $\Theta_2$ is false.
		\item $\textit{SAT}_{?[\phi_1]}(\theta_i) = \textit{SAT}_{?[\phi_2]}(\theta_i)$ for all $i=1\ldots m$.
	\end{enumerate}
	From which the claim follows.

	First, we can ignore all queries in which $\theta_i$ is unsatisfiable --- \emph{every} $\phi$ would yield $\textit{SAT}_{?[\phi]}(\theta_i) = \false$ in that case.
	Assume therefore that each $\theta_i$ has a satisfying assignment $v_i$.
	Take $\phi_1(\vec{y},\vec{x}) = \true$, and $\phi_2(\vec{y},\vec{x})$ that is true iff $\vec{y},\vec{x} = v_i$ for some $i$. %
	By this construction, $v_i \models \theta_i \land \phi_2$ for every $i$, implying $\textit{SAT}_{?[\phi_1]}(\theta_i) = \true$ for every $i$. The same of course also holds for $\phi_1$.

	It remains to show that $\Theta_2$ is false.
	Since $m < 2^n$ there is an assignment $\tilde{v}_x$ to $\vec{x}$ %
	that does not take part in any of the assignments $v_1,\ldots,v_m$; thus for every assignment of $\vec{y}$, the assignment $\tilde{v}_x$ to $\vec{x}$ yields the value $\false$ in $\theta_2$. Hence $\Theta_2$ is false.
	\end{proof}

	We now return to the setting of inference.
	The proof of \Cref{thm:query-nopoly} uses similar ideas to show the hardness of invariant inference for transition systems encoding $\qbf{2}$ instances. There are two complications beyond \Cref{thm:s2-s1-lower-bound}:
	\begin{inparaenum}[(i)]
			\item The queries the algorithm performs may be adaptive; and
			\item Hoare-queries on the transition relations $\horacle{\psigmatr{\phi}}{\alpha,\beta}$ are richer than SAT queries on the formulas $\textit{SAT}_{?[\phi]}(\alpha)$: the formula $\beta$ can query whether the valuations defined by $\alpha$ made $\phi$ true/false or a combination thereof (by referring to the flags), whereas the SAT query checks simply whether all valuations defined by $\alpha$ produced false for $\phi$.
	\end{inparaenum}
}
\begin{proof}[Proof of \Cref{thm:query-nopoly}]
Let $\A$ be a computationally unbounded Hoare-query algorithm.
We show that the number of Hoare queries performed by $\A$ on transition systems from $\hardclass$ with $\card{\voc}=n$ is at least $2^{\frac{n-1}{2}}$. %
To this end, we show that if $\A$ over $\card{\voc} = 2n+3$ %
performs less than $2^n$ queries, then there exist two formulas $\psi_1,\psi_2$ over $y_1,\ldots,y_n,x_1,\ldots,x_n$ such that %
\begin{itemize}
	\item all the Hoare queries performed by $\A$ on $\psigmatr{\psi_1}$ and $\psigmatr{\psi_2}$ (the transition relations of $\fts{\psi_1}$ and $\fts{\psi_2}$, respectively) return the same result, even though

	\item $\A$ should return different results when run on  $\fts{\psi_1} \in \ftset^n$ and $\fts{\psi_2}\in \ftset^n$, since $\fts{\psi_1}$ has an invariant in $\moncnf{2n+1}{}$ and $\fts{\psi_2}$ does not have an invariant (of any length). %
\end{itemize}
We begin with some notation. Running on input $\fts{\phi}$, we
abbreviate $\horacle{\psigmatr{\phi}}{\cdot, \cdot}$ by $\horacle{\phi}{\cdot, \cdot}$.
Denote the queries $\A$ performs and their results by $\horacle{\phi}{\alpha_1,\beta_1}=b_1,\ldots,\horacle{\phi}{\alpha_m,\beta_m}=b_m$.
We call an index $i$ \emph{sat} if $b_i = \false$.
We say that $\psi$ query-agrees with $\phi$ if $\horacle{\psi}{\alpha_i,\beta_i} = b_i$ for all $i$.
We say that $\psi$ sat-query-agrees with $\phi$ if for every $i$ such that $b_i = \false$ it holds that $\horacle{\psi}{\alpha_i,\beta_i} = \false$. %

We first find a formula $\phi$ over $y_1,\ldots,y_n,x_1,\ldots,x_n$ such that the sequence of queries $\A$ performs when executing on $\fts{\phi}$ is \emph{maximally satisfiable}: %
if $\psi$ sat-query-agrees with $\phi$, then $\psi$ (completely) query-agrees with $\phi$ on the queries, that is,
\begin{equation}
\forall \psi. \ \
	\left(\forall i. \ b_i=\false \Rightarrow \horacle{\psi}{\alpha_i,\beta_i}=b_i\right)
\implies
	\left(\forall i. \ \horacle{\psi}{\alpha_i,\beta_i}=b_i\right)
\end{equation}
	
We construct this sequence iteratively (and define $\phi$ accordingly) by always taking $\phi$ so that the result of the next query is $\false$ as long as this is possible while remaining consistent with the results of the previous queries:
Initially, choose some arbitrary $\phi^0$. At each point $i$, consider the first $i$ queries $\A$ performs on $\phi^i$, $\horacle{\phi^i}{\alpha_1,\beta_1}=b_1,\ldots,\horacle{\phi^i}{\alpha_i,\beta_i}=b_i$. If $\A$ terminates without performing another query, we are done: the desired $\phi$ is $\phi_i$.
Otherwise let $(\alpha_{i+1},\beta_{i+1})$ be the next query. Amongst formulas $\phi^{i+1}$ that query-agree on the first $i$ queries, namely, $\horacle{\phi^{i+1}}{\alpha_j,\beta_j} = b_j$ for all $j \leq i$, choose one such that $\horacle{\phi^{i+1}}{\alpha_{i+1},\beta_{i+1}} = \false$ if possible; if such $\phi^{i+1}$ does not exist take e.g.\ $\phi^{i+1} = \phi^{i}$. The dependency of $\A$ on $\phi^{i}$ is solely through the results of the queries to $\horacle{\psigmatr{\phi}}{\cdot,\cdot}$, so $\A$ performs the same $i$ initial queries when given $\phi^{i+1}$.
The result is a maximally satisfiable sequence, for if a formula $\psi$ differs in query $i+1$ in which the result is $\false$ instead of $\true$ we would have taken such a $\psi$ as $\phi^{i+1}$.

Let $\phi$ be such a formula with a maximally satisfiable sequence of queries $\A$ performs on $\phi$, $\horacle{\phi}{\alpha_1,\beta_1}=b_1,\ldots,\horacle{\phi}{\alpha_m,\beta_m}=b_m$.
For every sat $i$, take a counterexample $\sigma_i,\sigma'_i \models \alpha_i \land \psigmatr{\phi} \land \neg \beta'_i$. The single transition $(\sigma_i,\sigma'_i)$ of $\psigmatr{\phi}$ depends on the value of $\phi$ on at most one assignment to $\vec{x},\vec{y}$, %
so there exists a valuation $v_i: \vec{y} \cup \vec{x} \to \{\true,\false\}$ such that
\begin{equation}
\forall \psi. \ \ \psi(v_i) = \phi(v_i) \implies \sigma_i,\sigma'_i \models \alpha_i \land \psigmatr{\psi} \land \neg \beta'_i
\end{equation}
as well.
It follows that
\begin{equation}
\forall \psi. \ \ \psi(v_i) = \phi(v_i) \implies \horacle{\psi}{\alpha_i,\beta_i} = \horacle{\phi}{\alpha_i,\beta_i} = \false.
\end{equation}

Let $v_{i_1},\ldots,v_{i_t}$ be the valuations derived from the sat queries (concerning indexing, $b_i = \false$ iff $b_i = b_{i_j}$ for some $j$).
We say that a formula $\psi$ \emph{valuation-agrees} with $\phi$ on $v_{i_1},\ldots,v_{i_t}$ if $\psi(v_{i_j}) = \phi(v_{i_j})$ \emph{for all} $j$'s.
Since the sequence of queries is maximally satisfiable, if $\psi$ valuation-agrees with $\phi$ on $v_{i_1},\ldots,v_{i_t}$ then $\psi$ query-agrees with $\phi$, namely, $\horacle{\psi}{\alpha_i,\beta_i} = \horacle{\phi}{\alpha_i,\beta_i}$ for \emph{all} $i=1,\ldots,m$.
As the dependency of $\A$ on $\phi$ is solely through the results $b_1,\ldots,b_m$, it follows that $\A$ performs the same queries on $\psi$ as it does on $\phi$ and returns the same answer.

It remains to argue that if $m < 2^n$ then there exist two formulas $\psi_1,\psi_2$ that valuation-agree with $\phi$ on $v_{i_1},\ldots,v_{i_t}$ but differ in the correct result $\A$ should return: $\exists \vec{y}. \, \forall \vec{x}. \, \psi_1(\vec{y},\vec{x})$ is true, and so $\fts{\psi_1}$ has an invariant in $\moncnf{2n+1}{}$ (\Cref{lem:validity-vs-inference}), whereas $\exists \vec{y}. \, \forall \vec{x}. \, \psi_2(\vec{y},\vec{x})$ is not, and so $\fts{\psi_2}$ does not have an invariant of any length or form (\Cref{lem:validity-vs-inference}). %
This is possible because the number of constraints imposed by valuation-agreeing %
with $\phi$ on $v_{i_1},\ldots,v_{i_t}$ is less than the number of possible valuations of $\vec{x}$ for every valuation of $\vec{y}$ and vice versa:
\begin{equation}
	\psi_1(\vec{y},\vec{x}) =
		\bigwedge_{\substack{i=1..t \\ \theta(v_{i_j})=\false}}{
			(\vec{y},\vec{x}) \neq v_{i_j}
		}
\end{equation}
is true on all valuations except for some of $v_{i_1},\ldots,v_{i_t}$, and since $t \leq m < 2^n$ there exists some $\vec{y}$ such that for all $\vec{x}$, $(\vec{y},\vec{x})$ is not one of these valuations (recall that $\card{y} = n$ bits).
Dually,
\begin{equation}
	\psi_2(\vec{y},\vec{x}) =
		\bigvee_{\substack{i=1..t \\ \theta(v_{i_j})=\true}}{
			(\vec{y},\vec{x}) = v_{i_j}
		}
\end{equation}
is false on all valuations except for some of $v_{i_1},\ldots,v_{i_t}$, and since $t \leq m < 2^n$ for every $\vec{y}$ there exists $\vec{x}$ such that $(\vec{y},\vec{x})$ is not one of these valuations (recall that $\card{y} = n$ bits). This concludes the proof.
\end{proof}

\subsection{Extension to Interpolation-Based Algorithms}
\label{sec:interpolation-query-algorithms}
We now consider inference algorithms based on \emph{interpolation}, another operation supported by SAT solvers.
Interpolation has been introduced to invariant inference by~\citet{DBLP:conf/cav/McMillan03}, and since extended in many works (see~\Cref{sec:prelim-interpolation}).

Interpolation algorithms infer invariants from facts obtained with Bounded Model Checking (BMC), which we formalize as follows:
\begin{definition}[Bounded Reachabilitiy Check]
\label{def:extended-hoare-oracle}
The $k$-\emph{bounded reachability check} returns
	\begin{equation}
		\exthoracle{k}{\tr}{\alpha,\beta} \; \eqdef \; \alpha(\voc^0) \land \tr(\voc^0,\voc^1) \land \ldots \land \tr(\voc^{k-1},\voc^k) \implies \beta(\voc^k)
	\end{equation}
for $\alpha,\beta \in \Formulas{\Sigma}$,
where $\voc^0,\ldots,\voc^k$ are $k+1$ distinct copies of the vocabulary.
\end{definition}

\begin{definition}[Interpolation-Query Model] \label{def:interpolationQuery}
An \emph{interpolation-query oracle} is a query oracle $Q$ such that for every $\tr$, $\alpha,\beta \in \Formulas{\Sigma}$, and $k_1,k_2 \in \Nat$,
\begin{itemize}
	\item $\exthoraclet{Q}{k_1,k_2}{\tr}{\alpha,\beta} = \bot$ if $\exthoracle{k_1+k_2}{\tr}{\alpha,\beta} = \false$, and
	\item $\exthoraclet{Q}{k_1,k_2}{\tr}{\alpha,\beta} = \rho$ for $\rho \in \Formulas{\voc}$ such that $\exthoracle{k_1}{\tr}{\alpha,\rho}=\true$ and $\exthoracle{k_2}{\tr}{\rho,\beta}=\true$ otherwise.
\end{itemize}
We define $\itporaclesym$ to be the \emph{family} of all interpolation-query oracles.

An algorithm in the \emph{interpolation-query model}, also called an \emph{interpolation-query algorithm}, is an inference from queries algorithm
expecting any interpolation query oracle, where $k_1,k_2$ %
are bounded by a \emph{polynomial} in $n$ in all queries.
The query complexity in this model is %
$\querycomplexity_{\bbalg{}{}}^{\itporaclesym}(n)$.
\end{definition}
Interpolation-query oracles form a \emph{family} of oracles since different oracles can choose different $\rho$ for every $\tr,\alpha,\beta,k_1,k_2$.
Note that $\rho$ may be exponentially long. %

\OMIT{
	It starts with extending Hoare checks to Bounded Model Checking (BMC), which we formalize using extended Hoare queries:
	\begin{definition}[Extended Hoare-Query algorithm]
	An algorithm in the \emph{extended Hoare-query model} is an inference algorithm from extended Hoare queries, where the bound $k$ is bounded by a \emph{polynomial} in $n$.
	\sharon{so regarding my previous question, seems like you mean that there is an oracle for each $k$?}
	\end{definition}	
}

\subsubsection{Lower Bound on Interpolation-Query Algorithms}
We show an exponential lower bound on query complexity for interpolation-query algorithms.
To this end we prove the following adaptation of \Cref{thm:query-nopoly}:
\begin{theorem}
\label{thm:interpolation-nopoly}
Every interpolation-query inference algorithm, even computationally-unrestricted, deciding polynomial-length inference for the class of transition systems $\hardclass$ (\Cref{sec:hard-class}) %
has query complexity of $\explowerbound$. %
\end{theorem}

We remark that the lower bound on the interpolation-query model does not follow directly from the result for the Hoare-query model: an interpolant for $\exthoracle{k_1+k_2}{\tr}{\alpha,\beta}=\true$ depends on all traces of length $k_1+k_2$ starting from states satisfying $\alpha$, which may be an exponential number, so it cannot be computed simply by performing a polynomial number of Hoare queries to find these traces and computing an interpolant based on them. In principle, then, an interpolant can convey information beyond a polynomial number of Hoare queries.
Our proof argument is therefore more subtle: %
we show that there exists a choice of an interpolant that is not more informative than the existence of some interpolant (i.e., only reveals information on $\exthoracle{k_1+k_2}{\cdot}{\cdot,\cdot}$), in the special case of systems in $\hardclass$, in the maximally satisfiable branch of an algorithm's execution as used in the proof of \Cref{thm:query-nopoly}.
\iflong
\begin{proof}
Let $\A$ be a computationally unbounded algorithm using bounded reachability queries with bounds $k_1,k_2 < r(n)$ for some polynomial $r(n)$ (here $n=(\card{\voc}-1)/2$, as in the proof of \Cref{thm:query-nopoly}), with query complexity $m < \frac{2^{n}}{r(n)}$.

To prove the theorem it is convenient to first consider the algorithmic model which performs bounded-reachability queries of polynomial depth but not interpolation queries, intuitively performing BMC but without obtaining interpolants from the SAT solver.
Formally we consider \Cref{def:extended-hoare-oracle} as a query-oracle and first prove the lower bound for algorithms using this query oracle.
The proof follows the argument from the proof of \Cref{thm:query-nopoly}, relying on the fact that the BMC bounds $k_1,k_2 < r(n)$

\begin{sloppypar}
Assume that $\A$ performs only bounded reachability queries (without obtaining interpolants).
In what follows, we abbreviate $\exthoracle{k}{\psigmatr{\phi}}{\cdot,\cdot}$ by $\exthoracle{k}{\phi}{\cdot,\cdot}$.
We start the same as the proof of \Cref{thm:query-nopoly} to obtain a formula $\phi$ such that the sequence of bounded reachability queries $\A$ performs when executing on $\fts{\phi}$ is maximally satisfiable. In this proof this reads,
$
\forall \psi. \ \
	\left(\forall i. \ b_i=\false \Rightarrow \exthoracle{k_i}{\phi}{\alpha_i,\beta_i}=b_i\right)
\implies
	\left(\forall i. \ \exthoracle{k_i}{\phi}{\alpha_i,\beta_i}=b_i\right).
$
\end{sloppypar}

The main difference is that every sat query produces a counterexample \emph{trace} rather than a counterexample \emph{transition} as in \Cref{thm:query-nopoly}.
For every sat query $i$, we take a counterexample trace $\sigma^i_1,\ldots,\sigma^i_{k_i}$, namely, $\sigma^i_1 \models \alpha, \sigma^i_{k_i} \models \neg \beta$, and $\sigma^i_j, \sigma^i_{j+1} \models \psigmatr{\phi}$.
Every such transition $\sigma^i_j, \sigma^i_{j+1} \models \psigmatr{\phi}$ depends on at most one valuation of $\phi$.
Thus there exists valuations $v^i_1,\ldots,v^i_{k_i}$ so that every $\psi$ that valuation-agrees with $\phi$ on these valuations also allows the aforementioned counterexample trace, and thus $\exthoracle{k_i}{\psi}{\alpha_i,\beta_i}=\false$ as well.
As in the proof of \Cref{thm:query-nopoly}, it follows that if $\psi$ valuation-agrees with $\phi$ on all valuations $v^1_1,\ldots,v^1_{k_1},\ldots,v^m_1,\ldots,v^m_{k_m}$, then all queries on $\psi$ give that same result as those on $\phi$.
Since $k_1,\ldots,k_m < r(n)$, the number of these valuations is less that $r(n) \cdot m < 2^n$. %
The rest of the proof is exactly as in \Cref{thm:query-nopoly}, constructing two formulas $\psi_1,\psi_2$ valuation-agreeing with $\phi$ on all valuations, but one is true (in the $\qbf{2}$ sense) and the other is not.

We now turn to interpolation queries. Assume without loss of generality that every interpolation query is preceded by a bounded reachability query with the same bound, and that if the result is $\false$ the algorithm skips the interpolation query.

Consider the algorithm's execution on $\phi$ constructed above. We show that there exists an interpolation-query oracle that returns the same interpolants %
on the queries performed by the algorithm for both $\psi_1,\psi_2$ from above, and thus the algorithm's execution still does not distinguish between them.

Consider a point when the algorithms seeks an interpolant based on $\exthoracle{k_1+k_2}{\phi}{\alpha,\beta} = \true$. %
Let $S$ be the set of all formulas consistent with $\phi$ on all the valuations $v_i$ from above.
In particular, $\exthoracle{k_1+k_2}{\theta}{\alpha,\beta} = \true$ for all $\theta \in S$.
We construct a single interpolant valid for all $\theta \in S$, that is, a formula $\rho$ s.t.\ for every $\theta \in S$, $\exthoracle{k_1}{\theta}{\alpha,\rho}=\true$ and $\exthoracle{k_2}{\theta}{\rho,\beta}=\true$.
In particular, $\psi_1,\psi_2 \in S$, so this gives the desired interpolant for them.

Take $\hat{\tr} = \bigvee_{\theta \in S}{\psigmatr{\theta}}$.  %
We argue that $\exthoracle{k_1+k_2}{\hat{\tr}}{\alpha,\beta} = \true$.
For otherwise, there exists a trace $(\sigma_0,\ldots,\sigma_{k_1+k_2})$ %
of $\hat{\tr}$ such with $\sigma_0 \models \alpha$ and $\sigma_{k_1+k_2} \not\models \beta$. By the definition of $\hat{\tr}$, each transition $(\sigma_j,\sigma_{j+1})$ originates from $\psigmatr{\theta_j}$ for some $\theta_j \in S$. As before, this transition depends on the truth value of $\theta_j$ at one valuation $\hat{v}_j$ at the most. Furthermore, $\hat{v}_1,\ldots,\hat{v}_{k_1+k_2-1}$ are successive valuations, because all the transition systems in the class increment the valuation in the same way. Thus $\hat{v}_{j_1} \neq \hat{v}_{j_2}$.
Take a formula $\hat{\theta} \in S$ that agrees with $\theta_j$ on $\hat{v}_j$ for all $j=1,\ldots,k_1+k_2-1$; one exists because
\begin{enumerate}
	\item $\theta_j(\hat{v}_j)$ %
cannot contradict the valuations $\phi(v_i)$, because $\theta_i \in S$, and
	\item $\theta_{j_1}(\hat{v}_{j_1})$ does not contradict $\theta_{j_2}(\hat{v}_{j_2})$ for $\hat{v}_{j_1} \neq \hat{v}_{j_2}$ for $j_1 \neq j_2$.
\end{enumerate}
Thus $(\sigma_0,\ldots,\sigma_{k_1+k_2})$ is also a valid trace of $\hat{\theta} \in S$, which is a contradiction to $\exthoracle{k_1+k_2}{\hat{\theta}}{\alpha,\beta} = \true$.

Thus there exists some $\rho$ an interpolant for $\hat{\tr},\alpha,\beta,k_1,k_2$.
We choose the interpolation query oracle so that
\begin{equation*}
	\exthoraclet{Q}{k_1,k_2}{\theta}{\alpha,\beta} = \rho
\end{equation*}
for all $\theta \in S$.
This is a valid choice of interpolant: $\exthoracle{k_1}{\theta}{\alpha,\rho}=\true$ and $\exthoracle{k_2}{\theta}{\rho,\beta}=\true$ because $\exthoracle{k_1}{\hat{\tr}}{\alpha,\rho}=\true$ and $\exthoracle{k_2}{\hat{\tr}}{\rho,\beta}=\true$ and $\hat{\tr}$ includes all the transitions of $\psigmatr{\theta}$.

The claim follows.
\end{proof}
\else
The full proof appears in the extended version~\cite{extendedVersion}.
\fi   %
\subsection{Impossibility of Generalization from Partial Information}
\label{sec:gen-intro}

Algorithms such as PDR use generalization schemes to generalize from specific states to clauses (see \Cref{sec:generalization-partial-info-motivation} and \Cref{sec:prelim-pdr}).
It is folklore that ``good'' generalization is the key to successful invariant inference.
In this section, we apply the results of \Cref{sec:information-bounds} to shed light on the question of generalization. %
Technically, this is a discussion of the results in \Cref{sec:information-bounds}.

Clearly, if the generalization procedure has \emph{full information}, that is, has unrestricted access to the input---including the transition relation---then unrestricted computational power makes the problem of generalization trivial (as is every other problem!). %
For example, ``efficient'' inference can be achieved by a backward-reachability algorithm (see \Cref{sec:overview-gen-dilemma}) that blocks counterexamples through a generalization that uses clauses from a target invariant it can compute. %
This setting of full-information, computationally-unrestricted generalization was used by \citet{pldi/PadonMPSS16} in an interactive invariant inference scenario.
\label{sec:ideal-gen}

Our analysis in \Cref{sec:information-bounds} implies that the situation is drastically different when generalization possesses \emph{partial information}: the algorithm does not know the transition relation exactly, and only knows the results of a polynomial number of Hoare queries. By \Cref{thm:query-nopoly}, no choice of generalization made on the basis on this information can in general achieve inference in a polynomial number of steps.
This impossibility holds even when generalization uses unrestricted computational power, and thus it is a problem of \emph{information}. To further illustrate the idea of partial information, we note that the problem remains hard even when generalization %
is equipped with information beyond the results of a polynomial number of Hoare queries, information of the reachability of the transition system from $\Init$ and backwards from $\Bad$ in a \emph{polynomial} number of steps\footnote{
	This can be shown by noting that in $\hardclass$ (used to establish the exponential lower bound) such polynomial-reachability information can be obtained from a polynomial number of Hoare queries, reducing this scenario to the original setting.};
in contrast, information of the states reachable in \emph{any} number of steps constitutes full information and the problem is again trivial with unrestricted computational power.

Finally, the same challenge of partial information is present in algorithms basing generalization on a polynomial number of interpolation queries, as follows from \Cref{thm:interpolation-nopoly}.
\section{The Power of  Hoare-Queries}
\label{sec:rice-vs-ice}

Hoare queries are rich in the sense that the algorithm can choose a precondition $\alpha$ and postcondition $\beta$ and check $\horacle{\tr}{\alpha,\beta}$, where $\alpha$ may be different from $\beta$.
As such, algorithms in the Hoare-query model can utilize more flexible queries beyond querying for whether a candidate is inductive. %
In practice, this %
richer form of queries facilitates an incremental %
construction of invariants in %
complex syntactic forms. %
For example, PDR~\cite{ic3,pdr} incrementally learns clauses in different frames via relative inductiveness checks, and interpolation learns at each iteration a term of the invariant from an interpolant~\cite{DBLP:conf/cav/McMillan03} (see \Cref{sec:prelim-pdr}).
In this section we analyze this important aspect of the Hoare-query model and show that it can be strictly stronger than inference based solely on presenting whole candidate inductive invariants. We formalize the latter approach by the model of \emph{inductiveness-query algorithms}, closely related to ICE learning~\cite{ICELearning}, and construct a class of transition systems for which a simple Hoare-query algorithm can infer invariants in polynomial time, but every inductiveness-query algorithm requires an exponential number of queries.

\subsection{Inductiveness-Query Algorithms}
\label{sec:inductiveness-query-algorithms}
We define a more restricted model of invariant inference using only inductiveness queries.

\begin{definition}[Inductiveness-Query Model] \label{def:inductivenessQuery}
An \emph{inductiveness-query oracle} is a query oracle $Q$ such that for every $\tr$ and $\alpha \in \Formulas{\Sigma}$ satisfying $\Init \implies \alpha$ and $\alpha \implies \neg \Bad$,
\begin{itemize}
	\item $\indclet{Q}{\tr}{\alpha} = \true$ if $\alpha \land \tr \implies \alpha'$, and
	\item $\indclet{Q}{\tr}{\alpha} = (\sigma, \sigma')$ such that $(\sigma,\sigma') \models \alpha \land \tr \land \neg \alpha'$ otherwise.
\end{itemize}
We define $\indclesym$ to be the \emph{family} of all inductiveness-query oracles.

An algorithm in the \emph{inductiveness-query model}, also called an \emph{inductiveness-query algorithm}, is an inference from queries algorithm
expecting any inductiveness query oracle.
The query complexity in this model is %
$\querycomplexity_{\bbalg{}{}}^{\indclesym}(n)$.
\end{definition}

Inductiveness-query oracles form a \emph{family} of oracles since different oracles can choose different $(\sigma,\sigma')$ for every $\tr,\alpha$.
Accordingly, the query complexity of inductiveness-query algorithms is measured as a worst-case query complexity over all possible choices of an inductiveness-query oracle in the family.

\paragraph{ICE learning and inductiveness-queries}
The inductiveness-query model is closely related to ICE learning~\cite{ICELearning}, except here the learner is provided with full information on $\Init,\Bad$ instead of positive and negative examples (and the algorithm refrains from querying on candidates that do not include $\Init$ or do not exclude $\Bad)$.
This model captures several interesting algorithms (see~\Cref{sec:prelim-ice}).
Our complexity definition in the inductiveness-query model being the worst-case among all possible oracle responses is in line with the analysis of strong convergence in~\citet{ICELearning}. Hence, lower bounds on the query complexity in the inductiveness query model imply lower bounds for the strong convergence of ICE learning.
\begin{changebar}
We formalize this in the following lemma, using terminology borrowed from \citet{ICELearning} (see \Cref{sec:prelim-ice}):
\begin{lemma}
Let $\progs{}$ be a class of transition systems, and $\mathcal{L}$ a class of candidate invariants.
Assume that deciding the existence of an invariant in $\mathcal{L}$, given an instance from $\progs{}$, requires at least $r$ queries in the inductiveness-query model.
Then every strongly-convergent ICE-learner for $(\progs{},\mathcal{L})$ has round complexity at least $r$.
\end{lemma}
\end{changebar}
\iflong
\begin{proof}
Given a strongly-convergent ICE-learner $\A$ with round-complexity at most $r$, we construct an inductiveness-query algorithm for deciding $(\progs{},C)$ in at most $r$ queries, in the following way.
Simulate at most $r$ rounds of $\A$, and implement a teacher as follow:
When $\A$ produces a candidate $\theta \in C$,
\begin{itemize}
	\item Check that $\Init \implies \theta$, otherwise produce a positive example, a $\sigma$ s.t.\ $\sigma \models \Init, \sigma \not\models \theta$;
	\item Check that $\theta \implies \neg\Bad$, otherwise produce a negative example, a $\sigma$ s.t.\ $\sigma \models \Bad, \sigma \models \theta$;
	\item Perform an inductiveness query for $\theta$.
		  If $\theta$ is inductive, we are done---return \emph{true}. Otherwise, the inductiveness-query oracle produces a counterexample---pass it to $\A$.
\end{itemize}
If $r$ rounds did not produce an inductive invariant, return \emph{false}.

The teacher we implement always extends the learner's sample with an example that actually is an example in the target description (see \Cref{sec:prelim-ice}), and that rules out the current candidate. Thus, if there exists a correct $h \in C$, $\A$ finds one after at most $r$ iterations, and we return \emph{true}. Otherwise, we terminate after at most $r$ with the last candidate not an inductive invariant, and we return \emph{false}.
\end{proof}
\else
The proof appears in the extended version~\cite{extendedVersion}.
\fi

\paragraph{Inductiveness queries vs. Hoare queries}
Inductiveness queries are specific instances of Hoare queries, where the precondition and postcondition are the same. Since Hoare queries can also find a counterexample in a polynomial number of queries (\Cref{lem:hoare-cti}), inductiveness-query algorithms can be simulated by Hoare-query algorithms.
Our results in the rest of this section establish that the converse is not true.
\OMIT{
Inductiveness-queries can be simulated using Hoare-queries, in the following sense:
	\begin{lemma}
	\label{lem:ind-vs-hoare}
	If $\bbalg{\indclesym}{}(\Init,\Bad,n,\bb{\tr})$ decides length-parametric %
	in a polynomial number of inductiveness-queries, then there exists a Hoare-query algorithm $\bbalgt{\B}{\indclesym}{}(\Init,\Bad,n,\bb{\tr})$ deciding length-parametric inference %
	in a polynomial number of Hoare-queries.
	\end{lemma}
	\begin{proof}
	Upon query $\indcle{\tr}{\alpha}$ by $\A$, query $\horacle{\tr}{\alpha,\alpha}$, and return a counterexample via \Cref{lem:hoare-cti} when the answer is $\false$. The counterexamples $\B$ finds in this way constitute an inductiveness-query oracle, and thus the number of queries before $\A$ finds an invariant must be polynomial.
	\end{proof}	
}

\subsection{Separating Inductiveness-Queries from Hoare-Queries}
\label{sec:ind-vs-hoare-all}
In this section we show that the Hoare query model (\Cref{def:hoareQuery})  is strictly
stronger than the inductiveness query model (\Cref{def:inductivenessQuery}). %
\OMIT{
	We construct a class of transition systems and invariants %
	for which
	\begin{itemize}
		\item invariant inference \emph{can} be solved in polynomial query complexity %
	by a simple algorithm in the Hoare query model, PDR with one frame, which is not only efficient in its query complexity but also in its time complexity, whereas
		\item \emph{no} algorithm operating through inductiveness queries alone can infer an invariant in a polynomial number of queries.
	\end{itemize}
	To the best of our knowledge, this is the first lower bound on the number of queries required for ICE learning~\cite{ICELearning}.
	\sharon{no longer the first. rephrase}
}
We will prove the following main theorem:
\begin{theorem}
\label{thm:hoare-inductive-main-theorem}
There exists a class of systems $\monmax$ %
for which
\begin{itemize}
	\item polynomial-length invariant inference has \emph{polynomial} query complexity
in the \emph{Hoare-query model} %
(in fact, also polynomial time complexity %
modulo the query oracle), but
	\item every algorithm in the \emph{inductiveness-query model} requires an \emph{exponential} number of queries.
\end{itemize}
\end{theorem}
The upper bound is proved in \Cref{cor:monmax-upper}, and the lower bound in \Cref{cor:monmax-lower}.

\subsubsection{Maximal Transition Systems for Monotone Invariants}
We first define the transition systems with which we will prove
\Cref{thm:hoare-inductive-main-theorem}. %
We start with a definition:

\OMIT{

	\subsubsection{Monotone CNF Invariants}
	\sharon{they are not used as invariants yet. I wouldn't use "invariants" in the title}
	\begin{definition}[Monotone CNF Formula]
	  Let  $\cnf{n}{}$ be the set of propositional formulas in CNF having at most $n$ clauses.
	  Let $\moncnf{n}{}$ be the subset of $\cnf{n}{}$ in which all literals are negative.
	\end{definition}
}

\begin{definition}[Maximal System]
Let $\Init, \Bad \not\equiv \false$ and let $\varphi$ be a formula such that $\Init \implies \varphi$ and $\varphi \implies \lnot \Bad$.
The \emph{maximal transition system} w.r.t.\ $\varphi$ is ($\Init, \maxtr{\varphi}, \Bad$) where
\iflong
\begin{equation*}
		\maxtr{\varphi} = \varphi \rightarrow \varphi'.
\end{equation*}
\else
$\maxtr{\varphi} = \varphi \rightarrow \varphi'$. \linebreak
\fi
A maximal transition system is illustrated as follows:

\begin{center}
  \includegraphics[width=0.2\textwidth]{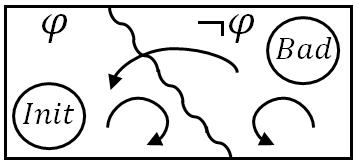}
\end{center}
\end{definition}

Note that $\maxtr{\varphi}$ goes from any state satisfying $\varphi$ to any state satisfying $\varphi$, and
from any state satisfying $\lnot \varphi$ to all states, good or bad.
$\maxtr{\varphi}$ is \emph{maximal} in the sense that it allows all transitions that do not violate the consecution of $\varphi$.
Thus any transition relation $\tilde{\tr}$ for which $\varphi$ satisfies consecution has $\tilde{\tr} \implies \maxtr{\varphi}$.

\begin{lemma}
\label{lem:maxsys-unique-inv}
A maximal transition system %
($\Init, \maxtr{\varphi}, \Bad$) has a unique inductive invariant, $\varphi$.
\end{lemma}

\begin{proof}
Let $I$ be any invariant of ($\Init, \maxtr{\varphi}, \Bad$).
By the definition of $\maxtr{\varphi}$ and the fact that  $\Init \implies \varphi$, the set of
states reachable from $\Init$ is exactly the set of states satisfying $\varphi$.  Thus
$\varphi \implies I$.

Since $\maxtr{\varphi}$ allows transitions from any state satisfying $\lnot \varphi$ to $\Bad$,
$I \implies \varphi$.
\end{proof}

The class of transition systems on which we focus, $\monmaxerror$, is the class of maximal systems for monotone invariants, $\monmaxnoerror$, together with certain unsafe systems.

Formally, %
for each $k \in \mathbb{N}$, we define $\monmaxnoerror^k$ as the class of all transition systems $(\Init_k, \maxtr{\varphi}, \Bad_k)$ for $\Init_k,\Bad_k$ from $\hardclass^k$ (\Cref{sec:hard-class}) and  $\varphi \in \moncnf{2k+1}{}$ such that $\Init_k \implies \varphi$ and $\varphi \implies \lnot \Bad_k$. We then define $\monmaxnoerror = \bigcup_{k \in \mathbb{N}} \monmaxnoerror^k$.
Further, for each $k$ we take the unsafe program $\badts{k} =  (\Init_k, \true, \Bad_k)$, and define the class $\monmaxerror = \monmaxnoerror \cup \{\badts{k} \mid k \in \mathbb{N}\}$.
Below we abbreviate and refer to the class $\monmaxerror$ as ``monotone maximal systems''.

Note that for each $k$, only a single transition system, $\badts{k}$, in $\monmaxerror$  does not have an invariant, and the others have a monotone invariant.
Still, \Cref{cor:monmax-lower} establishes a lower bound on polynomial-length inference for $\monmaxerror$  using inductiveness queries.
This means that using inductiveness queries alone, it is hard to distinguish between monotone invariants (otherwise decision would have been feasible via search).
On the other hand, with Hoare queries, search becomes feasible (establishing the upper bound).

\subsubsection{Upper Bound for Hoare-Query Algorithms for Monotone Maximal Systems}
\label{sec:monotone-upper}
A simple algorithm can find inductive invariants of monotone maximal systems with a polynomial number of queries. It is essentially PDR with a single frame.
The ability to \emph{find} invariants for $\monmax$ (and check invariants) shows that it is possible to \emph{decide} polynomial-length inference for $\monmaxerror$. %

We now present the PDR-1 algorithm (which was also discussed in \Cref{sec:generalization-partial-info-motivation}, and is cast here formally as a Hoare-query algorithm).
This is a backward-reachability algorithm, operating by repeatedly checking for the existence of a counterexample to induction, and obtaining a concrete example by the method discussed in \Cref{lem:hoare-cti}. The invariant is then strengthened by conjoining the candidate invariant with the negation of the formula \textsc{Block} returns. This formula is a subset of the cube of the pre-state.
In PDR-1, \textsc{Block} performs generalization by dropping a literal from the cube whenever the remaining conjunction does not hold for any state reachable in at most one step from $\Init$. The result is the strongest conjunction whose negation does not exclude any state reachable in at most one step. (This might exclude reachable states in general transition systems, but not in monotone maximal systems, since maximality ensures that their diameter is one.)

\begin{algorithm}
\caption{PDR-1 invariant inference in the Hoare-query model}
\label{alg:pdr-1}
\begin{algorithmic}[1]
\begin{footnotesize}
\Procedure{PDR-1}{\Init, \Bad, $\bb{\tr}$}	\qquad \qquad \quad \qquad // Backward-reachability with PDR-1 generalization
	\State $I \gets \neg\Bad$
	\While{$\horacle{\tr}{I,I} = \false$} \qquad \qquad \qquad \quad $\,$ // $I$ not inductive
		\State $\label{algln:cti}$ $(\sigma,\sigma') \gets \Call{model}{\bb{\tr}, I, \neg I'}$ \qquad \qquad \quad // counterexample to induction of $I$. implemented using \Cref{lem:hoare-cti}
		\State $\label{algln:min-ret}$ $d \gets$ \Call{Block-PDR-1}{\Init, $\bb{\tr},\sigma$}
		\State $I \gets I \land \neg d$
	\EndWhile
	\Return $I$
\EndProcedure
\\
\Procedure{Block-PDR-1}{\Init, $\bb{\tr}, \sigma$} \quad \ $\,$ \qquad \quad // Generalization according to one-step reachability
	\State $d \gets \cube{\sigma}$
	\For{$l \in \cube{\sigma}$}
		\State $t \gets d \setminus \set{l}$
		$\label{algln:minimize}$ \If{$\Init \implies \neg t \land \horacle{\tr}{\Init, \neg t}$} \qquad // $\Init \implies t \land \Init \land \tr \implies \neg t'$
			\State $d \gets t$
		\EndIf
	\EndFor
	\Return $d$
\EndProcedure
\end{footnotesize}
\end{algorithmic}
\end{algorithm}

The main property of monotone CNF formulas we exploit in the upper bound is the
ability to reconstruct them from \emph{prime consequences}. %
\begin{definition}[Prime Consequence]
A clause $c$ is a \emph{consequence} of $\varphi$ if $\varphi \implies c$.
A \emph{prime consequence}, $c$,  of $\varphi$ is a minimal consequence of $\varphi$, i.e., no proper subset
of $c$ is a consequence of $\varphi$.
\end{definition}

\begin{theorem}[Folklore]
\label{thm:montone-prime-consequence}
If $\varphi \in \moncnf{n}{}$ and a clause $c$ is a prime consequence of $\varphi$ then $c$ is a clause of $\varphi$.
\end{theorem}
\Cref{thm:montone-prime-consequence} is the dual of the folklore theorem on prime implicants of monotone DNF formulas as used e.g.\ by \citet{DBLP:journals/cacm/Valiant84}.
\iflong
\neil{this should be a basic, standard fact which you should not reprove in a research paper.}
For completeness, we provide a proof here:
\begin{proof}
Write $\varphi = c_1 \land \ldots \land c_n$.

First we argue that $c$ is monotone. This is because dropping all positive literals from $c$ also results in a consequence $\tilde{c}$ of $\varphi$; otherwise there is a valuation $v \models \varphi, v \not\models \tilde{c}$, but $v \models c$. Consider $\tilde{v}$ which is obtained by turning the literals in $c \setminus \tilde{c}$ to $\false$. Since $\varphi$ is monotone, $\tilde{v} \models \varphi$ still. As these literals appear positively in $c$, $\tilde{v} \not\models c$ the same way that $v \not\models \tilde{c}$. This is a contradiction to the premise.

Second, we argue that if $c'$ is monotone and is a consequence of $\varphi$, then there exists $c_i$ s.t.\ $c_i \subseteq c'$.
Otherwise, let $v$ be the valuation that assigns $\true$ to all variables in $c'$ and $\false$ to the rest. Clearly $v \not\models c'$. However, $v \models c_i$ for every $i$, since by our assumption there exists a literal $\neg x_i \in c_i \setminus c'$ to which $v$ assigns $\true$. Thus $v \models \varphi$. This is a contradiction to $c'$ being a consequence of $\varphi$.

The claim follows.
\end{proof}
\else
For completeness we provide a proof in the extended version~\cite{extendedVersion}.
\fi

We use this property to show that PDR-1 efficiently finds the invariants of the safe maximal monotone systems $\monmaxnoerror$, as implied by the following, slightly more general, lemma:
\begin{lemma}
\label{lem:pdr1-complexity}
Let $\TS = (\Init,\tr,\Bad)$ be a transition system over $\voc$, $n=\card{\voc}$, and $m \in \Nat$ %
such that
\begin{enumerate}[(i)]
	\item $\TS$ is safe,
	\item every reachable state in $\TS$ is reachable in at most one step from $\Init$,
	\item \iflong this set can be described by a $\moncnf{m}{}$ formula, namely, there is $\varphi \in \moncnf{m}{}$ such that $\sigma \models \varphi$ iff $\sigma \models \Init$ or $\exists \sigma_0 \in \Init$ s.t.\ $(\sigma_0,\sigma) \models \tr$.
\else  this set can be described by a formula $\varphi \in \moncnf{m}{}$.
\fi
\end{enumerate}
Then \mbox{\rm PDR-1}($\Init,\Bad,\bb{\tr}$) returns the inductive invariant $\varphi$ for %
$\TS$ with at most $n \cdot m$ Hoare queries. %
\end{lemma}
\iflong
\begin{proof}
Let $\varphi$ be as in the premise.
We show that $I$ of \Cref{alg:pdr-1} (1) always overapproximates $\maxinv{\varphi}$, (2) is strengthened with a new clause from $\maxinv{\varphi}$ in every iteration.

(1) We claim by induction on the number of iterations in \Cref{alg:pdr-1} that $\maxinv{\varphi} \implies I$. Clearly this holds initially.
In \cref{algln:cti}, $\sigma \not\models \maxinv{\varphi}$, otherwise $\sigma' \models \maxinv{\varphi}$, and $\maxinv{\varphi} \implies I$ gives $\sigma' \models I$ which is a contradiction to $(\sigma,\sigma')$ being a CTI.
Now, for every $\sigma \not\models \maxinv{\varphi}$, its minimization $\neg d$ is a consequence of $\maxinv{\varphi}$, namely, $\maxinv{\varphi} \implies \neg d$, because $\neg d$ holds for all states reachable in at most one step ($\Init \implies \neg d, \, \Init \land \tr \implies \neg d'$) and those are states satisfying $\maxinv{\varphi}$.
Thus also $\maxinv{\varphi} \implies I \land \neg d$.

(2) %
As we have argued earlier, $\neg d$ in \cref{algln:min-ret} is a consequence of $\maxinv{\varphi}$.
We refer to it as the clause $c = \neg d$.
We %
argue that $c$ is a clause of $\varphi$.
By \Cref{thm:montone-prime-consequence} is suffices to show that $c$ is a \emph{prime} consequence of $\varphi$, seeing that $\varphi$ is monotone.
Assume (for the sake of contradiction) that $\tilde{c} \subsetneq c$ is a consequence of $\varphi$. Consider the minimization procedure as it attempts to remove a literal $\tilde{l} \in (\neg c) \setminus (\neg \tilde{c})$ (\cref{algln:minimize}). This literal is not removed, so $\Init \land t$ or $\Init \land \tr \land t'$ is satisfiable at this point. So there is a state reachable in at most one step that satisfies $t$, which means that $\varphi \land t$ is satisfiable. %
From this point onwards, literals are only omitted from $t$, apart from $\tilde{l}$ that is resurrected; thus $\neg c \subseteq t \cup \set{\tilde{l}}$.
These are conjunctions, %
so $\varphi \land (\neg c \setminus \set{\tilde{l}})$ is satisfiable.
But this means that $\varphi \not\implies (c \setminus \set{\neg \tilde{l}})$.
In particular, it follows that $\varphi \not\implies \tilde{c}$ since $\tilde{c} \subseteq c \setminus \set{\neg \tilde{l}}$.
Thus $\tilde{c}$ is not a consequence of $\varphi$, which is a contradiction to the premise. Therefore $c$ must be a prime consequence of $\varphi$.

It remains to argue that $s$ in \cref{algln:min-ret} is not already present in $I$. But this is true because $\sigma \models s$, and $s \not\models I$. %

Overall, after at most $m$ such iterations, $I \implies \varphi$ from (2), and also $\varphi \implies I$ from (1). Thus $I \equiv \varphi$, which is indeed an inductive invariant (it captures exactly the reachable states).
Minimization performs $n=\card{\voc}$ queries, and so the total number of queries is at most $nm$.
\end{proof}
\else
\fi

From this lemma and the uniqueness of the invariants (\Cref{lem:maxsys-unique-inv}) the upper bound for $\monmax$ follows easily\iflong
\else
~(the proofs appear in the extended version~\cite{extendedVersion})
\fi:
\begin{corollary}
\label{cor:monmax-upper}
Polynomial-length invariant inference of $\monmax$ %
can be decided in a polynomial number of Hoare queries.
\end{corollary}
\iflong
\begin{proof}
Let $p(\cdot)$ be the polynomial dictating the target length in \Cref{def:bounded-inference}. %
The Hoare-query algorithm runs PDR-1 for $p(n) \cdot n$ queries. %
If PDR-1 does not terminate, return \emph{no}. Otherwise, it produces a candidate invariant $\psi$. If $\psi \not\in \cnf{p(n)}{}$, return \emph{no}. Perform another check for whether $\psi$ is indeed inductive: %
if it is inductive, return \emph{yes}, otherwise \emph{no}.

Correctness: let $(\Init,\tr,\Bad) \in \monmax$.
If there exists an inductive invariant in $\cnf{p(n)}{}$, it is unique, and is
 the formula $\maxinv{\varphi} \in \moncnf{p(n)}{}$ that characterizes the set of states reachable in at most one step. By \Cref{lem:pdr1-complexity} PDR-1 finds $\maxinv{\varphi}$ in a polynomial number of Hoare queries.
Otherwise, $(\Init,\tr,\Bad)$ is not safe, or its unique inductive invariant $\varphi \not\in \cnf{p(n)}{}$. In both cases PDR-1 cannot produce an inductive invariant in $\cnf{p(n)}{}$, and terminates/is prematurely terminated after $p(n) \cdot n$ Hoare queries.
\end{proof}
\else
\fi

\begin{remark}
The condition that the invariant is monotone in \Cref{lem:pdr1-complexity} can be relaxed to \emph{pseudo-monotonicity}: A formula $\varphi$ in CNF is \emph{pseudo-monotone} if no propositional variable appears in $\varphi$ both positively and negatively. Thus it can be made monotone by renaming variables. It still holds for a pseduo-monotone CNF $\varphi$ that a prime consequence is a clause of $\varphi$, and therefore PDR-1 successfully finds an invariant in
a polynomial number of Hoare queries also for the class of maximal systems for \emph{pseudo-}monotone invariants. %
\end{remark}

\subsubsection{Lower Bound for Inductiveness-Query Algorithms for Monotone Maximal Systems}
\label{sec:monotone-lower}

We now prove that every inductiveness-query algorithm for the class of monotone maximal systems requires exponential query complexity.
The main idea of the proof is that inductiveness-query algorithms are oblivious to adding transitions:

\begin{theorem}
\label{thm:ind-max-as-all}
Let $X,Y$ be sets of transition systems, such that $Y$ \emph{covers} the transition relations of $X$, that is,
for every $(\Init,\tr,\Bad) \in X$ %
there exists $(\Init, \hat{\tr}, \Bad) \in Y$ over the same vocabulary s.t.\
\iflong
\begin{enumerate}
\else
\begin{inparaenum}[(1)]
\fi
	\item \label{it:max-to-all-cond1} $\tr \implies \hat{\tr}$, and
	\item \label{it:max-to-all-cond2} if $(\Init,\tr,\Bad)$ has an inductive invariant in $\cnf{p(n)}{}$, then so does $(\Init, \hat{\tr}, \Bad)$.
\iflong
\end{enumerate}
\else
\end{inparaenum}
\fi
Then if $\A$ %
is an inductiveness-query algorithm for $Y$ %
with query complexity $t$, then $\A$ is also an inductiveness-query algorithm for $X$ with query complexity $t$.
\end{theorem}
\begin{proof}
Let $\A$ be an algorithm for $Y$ %
as in the premise. We show that $\A$ also solves the problem for $X$.
Let $(\Init,\tr,\Bad) \in X$ %
and analyze $\bbalgt{\A}{Q}{}(\Init,\Bad,\bb{\tr})$, where $Q$ is some inductiveness-query oracle. %
Consider the first $t$ candidates, $\alpha_1,\ldots,\alpha_t$.
If one of them is an inductive invariant for $(\Init,\tr,\Bad)$, we are done (recall that the inductiveness query is only defined for queries with $\alpha_i$ s.t.\ $\Init \implies \alpha_i$ and $\alpha_i \implies \neg\Bad)$.
If we are not done, let $(\Init,\hat{\tr},\Bad) \in Y$ as in the premise for the given $(\Init,\tr,\Bad)$. %
We show that in this case $\bbalgt{\A}{Q}{}(\Init,\Bad,\bb{\tr})$ simulates $\bbalgt{\A}{Q'}{}(\Init,\Bad,\bb{\hat{\tr}})$ where $Q'$ is an inductiveness-query oracle derived from $Q$ by
$
\indclet{Q'}{\hat{\tr}}{\alpha_i}
		\; = \;
			\indclet{Q}{\tr}{\alpha_i}
$
for all $i=1,\ldots,t$.
Note that $\indclet{Q'}{\hat{\tr}}{\cdot}$ is a valid inductiveness-query oracle: by the assumption that $\alpha_i$ is not inductive for $\tr$, $\indclet{Q}{\tr}{\alpha} = (\sigma,\sigma')$, that is, $\sigma,\sigma' \models \alpha \land \tr \land \lnot \alpha'$.
From %
condition \ref{it:max-to-all-cond1}, %
$\tr \implies \hat{\tr}$, and so we deduce that also $\sigma,\sigma' \models \alpha \land \hat{\tr} \land \lnot \alpha'$.
Therefore, after at most $t$ queries, $\bbalgt{\A}{Q'}{}(\Init,\Bad,\bb{\hat{\tr}})$ terminates, returning either
\begin{inparaenum}[(i)]
	\item an inductive invariant %
	$\varphi \in \cnf{p(n)}{}$ for $(\Init,\hat{\tr},\Bad)$, which is also an inductive invariant for $(\Init,\tr,\Bad)$, by condition \ref{it:max-to-all-cond1}; or
	\item no inductive invariant in %
	$\cnf{p(n)}{}$ for $(\Init,\hat{\tr},\Bad)$, in which case this is also true for $(\Init,\tr,\Bad)$, by condition \ref{it:max-to-all-cond2}.
\end{inparaenum}
Either way $\bbalgt{\A}{Q}{}(\Init,\Bad,\bb{\tr})$ is correct and uses $\leq t$ queries.
\end{proof}

The lower bound for monotone maximal systems results from \Cref{thm:ind-max-as-all} together with the hardness previously obtained in \Cref{thm:query-nopoly-mon}:
\begin{corollary}
\label{cor:monmax-lower}
\label{thm:ind-max-lower}
Every inductiveness-query algorithm, even computationally-unrestricted, deciding polynomial-length inference for $\monmaxerror$ %
has query complexity of $\explowerbound$.
\sharon{the problem for which we show a lower bound here is actually invariant inference for $\cnf{n}{}$ and not arbitrary $p(n)$. We can inflate the vocabulary to get to other polynomilas, but this changes the class and its no longer $\monmaxerror$. But the upper bound still holds, so it still demonstrates the gap for every polynomial (thus the formulation of the main theorem is correct, just not necessarily with the same class)}
\end{corollary}
\iflong
\begin{proof}
For $\hardclass$ with invariant class $\moncnf{n}{}$, %
an exponential number of Hoare queries is necessary, by \Cref{thm:query-nopoly-mon}.
It follows that in the inductiveness-query model, an exponential query complexity is also required (since a Hoare-query algorithm can implement a valid inductiveness-query oracle).
By arguing that $\monmaxerror$ covers $\hardclass$ we can
apply \Cref{thm:ind-max-as-all} to deduce that $\monmaxerror$ with invariant class $\moncnf{n}{}$ %
also necessitates an exponential number of inductiveness queries:

Let $(\Init_k,\tr,\Bad_k) \in \hardclass$. Recall that in these systems, $n = 2k+3$ (the vocabulary size). %
If the system does not have an inductive invariant in $\moncnf{n}{}$, then $\badts{k} = (\Init_k,\true,\Bad_k) \in \monmaxerror$ satisfies the conditions of \Cref{thm:ind-max-as-all} (condition \ref{it:max-to-all-cond1} holds as evidently $\tr \implies \true$, and condition \ref{it:max-to-all-cond2} holds vacuously).
Otherwise, there exists an inductive invariant $\varphi \in \moncnf{n}{}$ for $(\Init_k,\tr,\Bad_k)$. In this case, the system $(\Init_k,\maxtr{\varphi},\Bad_k)$ satisfies the conditions of \Cref{thm:ind-max-as-all}: condition \ref{it:max-to-all-cond1} is due to the maximality of $\maxtr{\varphi}$, and \ref{it:max-to-all-cond2} holds as $\varphi$ is an inductive invariant.

Thus $\monmaxerror$ with the class $\moncnf{n}{}$ requires an exponential number of inductiveness queries.
Since $\monmaxerror$ has a monotone invariant or none at all, it follows that an exponential number inductiveness queries is also required for $\monmaxerror$ with
$\cnf{n}{}$, as desired.
\end{proof}
\else
The proof applies \Cref{thm:ind-max-as-all} to $\monmax$, which covers $\hardclass$, while the hardness of the latter was established in \Cref{thm:query-nopoly}. The full proof appears in the extended version~\cite{extendedVersion}.
\fi

We note that the transition relations in $\monmax$ are themselves polynomial in $\card{\voc}$. Hence the query complexity in this lower bound is exponential not only in $\card{\voc}$ but also in $\card{\tr}$ (see~\Cref{rem:short-tr}).

\begin{changebar}
Finally, it is interesting to notice that the safe systems in $\monmax$ have a \emph{unique} inductive invariant, and still the problem is hard.
\end{changebar}

\OMIT{
	\begin{theorem}
	A Hoare-query algorithm deciding polynomial-length inference for $\invtswbad{\moncnf{n}{}}$ requires an exponential query complexity.
	\sharon{you only refer to the transition systems and not to the class of invariants. Do you mean that the language of invariants includes only monotone CNF? do you consider all CNF? Here it matters}
	\end{theorem}
	\begin{proof}
	This follows directly from \Cref{thm:hoare-nopoly}, since this class contains all the transition systems in the proof of \Cref{thm:hoare-nopoly}: their invariants are in fact monotone. \yotam{probably good to put this strengthening (monotone-CNF, not just CNF) somewhere also in the previous section}
	\end{proof}
	Using \Cref{lem:ind-vs-hoare} it immediately follows that this class is hard for inductiveness-query algorithms:
	\begin{corollary}
	\label{lem:ind-all-mon}
	\sharon{inline in the proof}
	An inductiveness-query algorithm deciding polynomial-length inference for $\invtswbad{\moncnf{n}{}}$ \sharon{with which class of invariants?} requires an exponential query complexity.
	\end{corollary}
	The key point is that for inductiveness query algorithms, maximal transition systems are as hard as all transition systems with the same invariants.	
}

\section{Invariant Learning \& Concept Learning with Queries}
\label{sec:concept-vs-invariant}
The theory of \emph{exact concept learning}~\cite{DBLP:journals/ml/Angluin87} asks a learner to identify an unknown formula\footnote{In general, a concept is a set of elements; here we focus on logical concepts.} $\varphi$ from a class $\mathcal{L}$ using queries posed to a teacher. Prominent types of queries include \emph{membership}---given state $\sigma$, return whether $\sigma \models \varphi$---and \emph{equivalence}---given $\theta$, return true if $\theta \equiv \varphi$ or, otherwise, a counterexample, a $\sigma$ s.t.\ $\sigma \not\models \theta, \sigma \models \varphi$ or vice versa.

What are the connections and differences between concept learning \emph{formulas} in $\mathcal{L}$ and learning \emph{invariants} in $\mathcal{L}$? Can concept learning algorithms be translated to inference algorithms? These questions have spurred much research~\cite[e.g.][]{ICELearning,DBLP:journals/acta/JhaS17}.
In this section we study these questions with the tool of query complexity and our aforementioned results.

The most significant outcome of this analysis is a new hardness result (\Cref{thm:ice-vs-concept}) showing that ICE-learning is provably harder than classical learning: namely, that, as advocated by~\citet{ICELearning}, learning from counterexamples to induction is inherently harder than learning from examples labeled positive or negative.
The proof of this result builds on the lower bound of \Cref{cor:monmax-lower}.
We also establish (im)possibility results for directly applying algorithms from concept learning to invariant inference.

\para{Complexity: the easy, the complex, and the even-more-complex}
\begin{table}
\begin{footnotesize}
  \centering
  \caption{Concept vs.\ invariant learning: query complexity of learning $\moncnf{n}{}$}
  \begin{threeparttable}
    \begin{tabular}{l|l|l||l|l}
    \multicolumn{3}{c||}{Invariant Inference} & \multicolumn{2}{c}{Concept Learning} \\
    \hline\hline
          & Maximal Systems & General Systems &       &  \\
    \hline
    Inductiveness & \begin{tabular}{@{}l@{}}Exponential \\ (\Cref{cor:monmax-lower})\end{tabular} & \begin{tabular}{@{}l@{}}Exponential \\ (\Cref{thm:query-nopoly})\end{tabular} & Equivalence & \begin{tabular}{@{}l@{}} \ Subexponential\tnote{1} \, / \, Polynomial\tnote{2} \\ \cite{DBLP:journals/jmlr/HellersteinKSS12,DBLP:journals/ml/Angluin87}\end{tabular} \\
    \hline
    Hoare & \begin{tabular}{@{}l@{}}Polynomial \\ (\Cref{cor:monmax-upper})\end{tabular} & \begin{tabular}{@{}l@{}}Exponential \\ (\Cref{thm:query-nopoly})\end{tabular} & \begin{tabular}{@{}l@{}}Equivalence\,+\\Membership\end{tabular} & \begin{tabular}{@{}l@{}}Polynomial\\\cite{DBLP:journals/ml/Angluin87}\end{tabular} \\
    \end{tabular}%
  \begin{tablenotes}
  \item[1] proper learning \item[2] with exponentially long candidates
  \end{tablenotes}
  \end{threeparttable}
  \label{tab:concept-vs-results}%
\end{footnotesize}
\end{table}%
 In this paper we have studied the complexity of inferring $\mathcal{L}=\moncnf{n}{}$ invariants using Hoare/inductiveness queries in two settings: for \emph{general systems} (in \Cref{sec:information-bounds}), and for \emph{maximal systems} in \Cref{sec:rice-vs-ice}. \Cref{tab:concept-vs-results} summarizes our results and contrasts them with known complexity results in classical concept learning for the same class of formulas.
For the sake of the comparison, the table maps inductiveness queries to equivalence queries (as these are similar at first sight) and maps the more powerful setting of Hoare queries to the more powerful setting of equivalence and membership queries.

Starting with similarity, the gap in the complexity between Hoare- and inductiveness-queries in learning invariants for maximal systems parallels the gap between equivalence and equivalence + membership queries in concept learning. Our proof for the upper bound for Hoare queries is related to the upper bound in concept learning and simulations of concept learning algorithms (see below), but the lower bound for inductiveness queries uses very different ideas, and establishes stronger lower bounds than possible in concept learning, as we describe below. %

The similarity ends here.
First, %
{\textit{general systems are harder}}, and inferring $\mathcal{L}=\moncnf{n}{}$ invariants for them is harder than concept learning with the same $\mathcal{L}$, even with the full power of Hoare queries. This, unsurprisingly, illustrates the challenges stemming from transition systems with complex reachability patterns, such as a large diameter.
Second, even the hard cases for concept learning have \emph{lower} complexity than the hard invariant inference problems: learning concepts in $\mathcal{L}=\moncnf{n}{}$ has \emph{subexponential} query complexity (or even polynomial complexity when exponentially-long candidates are allowed), whereas we prove \emph{exponential} lower bounds for inference.
One important instance of this discrepancy shows that inductiveness queries are inherently \emph{weaker} than equivalence queries, as
learning $\moncnf{n}{}$ invariants in the \emph{inductiveness} model is \emph{harder} than learning $\moncnf{n}{}$ formulas using \emph{equivalence} queries.
Put differently, this is a hardness result for concept learning with \emph{ICE-equivalence queries}, which are like equivalence queries, only when the given $\theta$ is not equivalent to the target concept $\varphi$ the teacher responds with an \emph{implication} counterexample~\cite{ICELearning}: a pair $\sigma,\sigma'$ s.t.\ $\sigma \models \theta$ and $\sigma' \not\models \theta$, but $\sigma \not\models \varphi$ or $\sigma' \models \varphi$. Our results thus imply:
\begin{corollary}
\label{thm:ice-vs-concept}
There exists a class of formulas $\mathcal{L}$ that can be learned using a subexponential number of equivalence queries, but requires an exponential number of ICE-equivalence queries.
\end{corollary}
This result quantitatively corroborates the difference between %
{\textit{counterexamples to induction}} and %
{\textit{examples labeled positive or negative}}, a distinction advocated by~\citet[][]{ICELearning}.

The higher complexity of inferring invariants
has consequences for the feasibility of simulating queries (and algorithms) from concept learning in invariant inference, as we discuss next.

\para{Queries: some unimplementable algorithms}
\begin{table}
  \footnotesize
  \centering
  \caption{Concept vs.\ invariant learning: implementability of concept-learning queries}
    \begin{tabular}{c|c|c|c|c}
          & \multicolumn{2}{c|}{Maximal Systems} & \multicolumn{2}{c}{General Systems} \\
          \hline
          & Inductiveness & Hoare & Inductiveness & Hoare \\
          \hline
    Equivalence & \xmark     & \vmark     & \xmark     & \xmark \\
    Membership & \xmark     & \vmark     & \xmark     & \xmark \\
    \end{tabular}%
  \label{tab:concept-vs-queries}%
\end{table}
\Cref{tab:concept-vs-queries} summarizes our results for the possibility and impossibility of simulating concept learning algorithms in invariant learning.
This table depicts implementability (\vmark) or unimplementability (\xmark) of membership and equivalence queries used in concept learning a class of formulas $\mathcal{L}$ through inductiveness and
Hoare queries used in learning invariants for
{maximal systems} over $\mathcal{L}$, and
for {general systems} with candidate invariants in $\mathcal{L}$.
The proofs of impossibilities are based on the differences in complexity described above:
that neither equivalence nor membership queries can be simulated over general systems using even Hoare queries is implied by the hardness of general systems;
that neither equivalence nor membership can be simulated even over maximal systems using inductiveness queries is implied by the higher complexity of these compared to concept learning.
The only possibility result is of simulating inductiveness and membership queries using Hoare queries over maximal systems; the idea is that a Hoare query $\horacle{\maxtr{\varphi}}{\Init,\neg \cube{\sigma}}\overset{?}{=}\false$ implements a membership query on $\sigma$, thanks to fact that the inductive invariant is exactly the set of states reachable in one step, and that a membership query can disambiguate a counterexample to induction into a labeled example, so it is possible to simulate an equivalence query by an inductiveness query.
Interestingly, the algorithm we use to show the polynomial upper bound on Hoare queries for maximal systems,
PDR-1, can be obtained as such a translation of an algorithm from \citet{DBLP:journals/ml/Angluin87} performing concept learning of $\moncnf{n}{}$ using equivalence and membership queries.
\section{Related Work}
\label{sec:related-work}

\para{Complexity of invariant inference}
The fundamental question of the complexity of invariant inference in propositional logic has been studied by \citet{DBLP:conf/cade/LahiriQ09}. They show that deciding whether an invariant exists is PSPACE-complete. This includes systems with only exponentially-long invariants, which are inherently beyond reach for algorithms aiming to \emph{construct} an invariant. In this paper we focus on the search for polynomially-long invariants.
\citet{DBLP:conf/cade/LahiriQ09} study the related problem of template-based inference, and show it is $\psigma{2}$-complete. Polynomial-length inference for CNF formulas can be encoded as specific instances of template-based inference; the $\psigma{2}$-hardness proof of \citet{DBLP:conf/cade/LahiriQ09} uses more general templates and therefore does not directly imply the same hardness for polynomial-length inference.
The same work also shows that inference is only $\ppi{1}=\coNP$-complete when candidates are only conjunctions (or, dually, disjunctions). In this paper we focus on the richer class of CNF invariants.

\OMIT{
\paragraph{Decidability of invariant inference}
The decidability of the safety problem has been studied in many settings~\cite{?}.
\yotam{copied from intro}
\citet{DBLP:conf/popl/PadonISKS16} analyze classes of programs and candidate invariants where the safety problem is undecidable, and consider the decidability of finding an invariant in specific syntactic classes in these cases. In this paper we also take into account the length of the target invariant, leading to the study of complexity rather than of decidability.	
}

\OMIT{
\paragraph{Template-based invariant inference}
Many works employ syntactical templates for invariants, used to constrain the search~\cite[e.g.][]{DBLP:conf/cav/ColonSS03,DBLP:conf/sas/SankaranarayananSM04,DBLP:conf/pldi/SrivastavaG09,DBLP:journals/sttt/SrivastavaGF13,DBLP:series/natosec/AlurBDF0JKMMRSSSSTU15}.
In the definition of polynomial-length invariant inference (\Cref{def:bounded-inference}) we formalize the problem addressed by recent invariant inference algorithms that search for invariants in rich syntactical forms, with the motivation of achieving generality of the verification method and potentially improving the success rate.
\yotam{repetition} \citet{DBLP:conf/cade/LahiriQ09} study the problem with general templates and show it is $\psigma{2}$ complete. Their hardness proof does not directly imply the same hardness for polynomial-length inference.
}

\para{Black-box invariant inference}
Black-box access to the program in its analysis is widespread in research on testing~\cite[e.g.][]{blackboxtesting}.
In invariant inference, Daikon~\cite{DBLP:journals/tse/ErnstCGN01} initiated the black-box learning of \emph{likely} program invariants~\cite[see e.g.][]{DBLP:conf/icse/CsallnerTS08,DBLP:conf/issta/SankaranarayananCIG08}.
In this paper we are interested in inferring necessarily correct inductive invariants.
The ICE learning model, introduced by~\citet{ICELearning,DBLP:conf/popl/0001NMR16}, and extended to general Constrained Horn Clauses in later work~\cite{DBLP:journals/pacmpl/EzudheenND0M18}, pioneered a black-box view of inference algorithms such as Houdini~\cite{DBLP:conf/fm/FlanaganL01} and symbolic abstraction~\cite{DBLP:conf/vmcai/RepsSY04,DBLP:journals/entcs/ThakurLLR15}. %
The inductiveness model in our work is inspired by this work, focusing on black-box access to the transition relation while providing the learner with full knowledge of the set of initial and bad states.
Capturing PDR in a black-box model was achieved by extending ICE with relative-inductiveness queries~\cite{DBLP:conf/vmcai/VizelGSM17}. Our work shows that an extension is necessary, and applies to any Hoare-query algorithm.

\para{Lower bounds for black-box inference}
To the best of our knowledge, our work provides the first unconditional exponential lower bound for rich black-box inference models such as the Hoare-query model. An impossibility result for ICE learning in polynomial time in the setting of quantified invariants was obtained by~\citet{ICELearning}, based on the lower bound of~\citet{DBLP:journals/ml/Angluin90} for concept learning DFAs with equivalence queries. %
Our lower bound for monotone maximal systems
\begin{inparaenum}[(i)]
	\item demonstrates an exponential gap between ICE learning and Hoare-query algorithms such as PDR (\Cref{sec:rice-vs-ice}), and
	\item separates ICE learning from concept learning (\Cref{sec:concept-vs-invariant}); in particular, it holds even
	when candidates may be exponentially long (see~\Cref{thm:ind-max-lower} and~\citet[][Appendix B]{ICELearning-techreport}).
\end{inparaenum}
\OMIT{
	\citet{DBLP:journals/ml/Angluin87} introduced the model of exact concept learning with query types such as membership and equivalence, spurring a long line of study~\cite[e.g.][]{DBLP:journals/tcs/Angluin04,DBLP:conf/colt/Angluin89,DBLP:journals/ml/Angluin90,DBLP:journals/ipl/Bshouty96,DBLP:journals/jcss/BshoutyCGKT96,DBLP:conf/colt/BshoutyDVY17}. Notable applications in formal methods include the use of the $L^{*}$ algorithm~\cite{DBLP:journals/iandc/Angluin87} for synthesizing assumptions and guarantees in compositional reasoning~\cite{DBLP:conf/cav/AlurMN05,DBLP:conf/popl/AlurCMN05}.
	In this work we study a different notion of exact learning, learning \emph{inductive invariants} using \emph{Hoare queries} and related SAT-based checks. As we show, this differs from concept learning the invariants directly; for instance, monotone CNF formulas can be efficiently concept-learned, whereas transition systems with monotone CNF invariants are already hard (\Cref{thm:query-nopoly}), and again easy in maximal systems (\Cref{cor:monmax-upper}).
	Lower bounds techniques based on combinatorial dimensions have proved valuable in exact concept learning (see~\citet{DBLP:journals/tcs/Angluin04} for a survey); our lower bounds are based on different ideas. One notable difference is that our lower bounds apply even when queries may use exponentially long formulas---the lower bound on learning monotone invariants with inductiveness queries is in this sense stronger than the classical lower bound on learning monotone formulas using equivalence queries~\cite{DBLP:journals/ml/Angluin87,DBLP:journals/ml/Angluin90}. An interesting direction is to whether the assumption that transition relations are of polynomial length (see \Cref{sec:poly-tr}) could be advantageous in practice and incorporated into our theory.
}

\para{Learning and synthesis with queries}
Connections with exact learning with queries~\cite{DBLP:journals/ml/Angluin87} are discussed in \Cref{sec:concept-vs-invariant}.
The lens of synthesis has inspired many works applying ideas from machine learning to invariant inference~\cite[e.g.][]{DBLP:conf/icse/JhaGST10,DBLP:conf/sas/0001GHAN13,DBLP:conf/cav/SharmaNA12,DBLP:conf/esop/0001GHALN13,ICELearning,DBLP:journals/fmsd/SharmaA16}.
The role of learning with queries is recognized in prominent synthesis approaches such as Counterexample-Guided Inductive Synthesis (CEGIS)~\cite{DBLP:conf/asplos/Solar-LezamaTBSS06} and synthesizer-driven approaches~\cite[e.g.][]{DBLP:conf/synasc/Gulwani12,DBLP:conf/icse/JhaGST10,DBLP:journals/corr/LePPRUG17}, which learn from equivalence and membership queries~\cite{DBLP:journals/acta/JhaS17,DBLP:series/natosec/AlurBDF0JKMMRSSSSTU15,DBLP:conf/colt/BshoutyDVY17,DBLP:conf/cav/Drachsler-Cohen17}.
The theory of oracle-guided inductive synthesis~\cite{DBLP:journals/acta/JhaS17} theoretically studies the convergence of CEGIS in infinite concept classes using different types of counterexamples-oracles, and relates the finite case to the teaching dimension~\cite{DBLP:journals/jcss/GoldmanK95}.
In this work we study inference based on a different form of queries, and prove lower bounds on the convergence rate in finite classes.

\para{Proof complexity}
Proof complexity studies the power of polynomially-long proofs in different proof systems. A seminal result is that a propositional encoding of the pigeonhole principle has no polynomial resolution proofs~\cite{DBLP:journals/tcs/Haken85}. Ideas and tools from proof complexity have been applied to study SAT solvers~\cite[e.g.][]{DBLP:journals/ai/PipatsrisawatD11} and recently also SMT~\cite{DBLP:conf/cav/RobereKG18}. Proof complexity is an alternative technical approach to study the complexity of proof search algorithms, by showing that some instances do not have a short proof, showing a lower bound regardless of how search is conducted.
Our work, inspired by learning theory, provides exponential lower bounds on query-based search even when the proof system is sufficiently strong to admit short proofs: in our setting, there is always a short derivation of an inductive invariant by generalization in backward-reachability, blocking counterexamples with the optimal choice, using clauses from a target invariant (see \Cref{sec:ideal-gen}).
We expect that proof complexity methods would prove valuable in further study of inference.  
\vspace{-0.1cm}
\section{Conclusion}
Motivated by the rise of SAT-based invariant inference algorithms, we have attempted to elucidate some of the principles on which they are based by a theoretical complexity analysis of algorithms attempting to infer invariants of polynomial size.
We have developed information-based analysis tools, inspired by machine learning theory, to investigate two focal points in SAT-based inference design:
\begin{inparaenum}
	\item \emph{Generalization}, which we have shown to be impossible from a polynomial number of Hoare queries in the general case;
	\item \emph{Rich Hoare queries}, beyond presenting candidate invariants, which we have shown to be pivotal in some cases.
\end{inparaenum}
Our upper bound for PDR on the class of monotone maximal systems is a first step towards theoretical conditions guaranteeing polynomial running time for such algorithms.
\begin{changebar} 
One lesson from our results is the importance of characteristics of the transition relations (rather than of candidate invariants), which make the difference between the lower bound for general systems (\Cref{thm:query-nopoly}) and the upper bound for maximal systems (\Cref{cor:monmax-upper}), both for the same class of candidate invariants. 
\end{changebar}We believe that theoretical guarantees of efficient inference would involve special classes of transitions systems and algorithms using repeated generalization employing rich Hoare queries. %

At the heart of our analysis lies the observation that many interesting SAT-based algorithms can be cast in a black-box model.
This work focuses on %
the limits and opportunities in black-box inference and shows interesting information-theoretic lower bounds.
One avenue for further research is an information-based analysis of black-box models extended with white-box capabilities, e.g.\ by investigating syntactical conditions on the transition relation that simplify generalization. %

\section*{Acknowledgments} 
\iflong
\else{\small 
\fi
We thank our shepherd and the anonymous referees for comments that improved the paper. We thank Kalev Alpernas, %
Nikolaj Bj{\o}rner, 
P.\ Madhusudan, Yishay Mansour, Oded Padon, Hila Peleg, Muli Safra, and James~R.~Wilcox for insightful discussions and suggestions, and Gil Buchbinder for saving a day.
The research leading to these results has received funding from the European Research Council under the European Union's Horizon 2020 research and innovation programme (grant agreement No [759102-SVIS]).
This research was partially supported by the National Science Foundation (NSF) grant no. CCF-1617498, by Len Blavatnik and the Blavatnik Family foundation,
the Blavatnik Interdisciplinary Cyber Research Center, Tel Aviv University,
the United States-Israel Binational Science Foundation (BSF) grant No.\ 2016260,
and the Israeli Science Foundation (ISF) grant No.\ 1810/18.
\iflong
\else
}
\fi 

\clearpage
\bibliography{refs}

\iflong
\clearpage
\appendix

\section{Duality of Backward- \& Forward- Reachability}
\label{sec:forward-backward-duality}
Throughout the paper we study %
invariant inference w.r.t.\
CNF formulas with $p(n)$ clauses (where $n= \card{\voc}$). %
Our results apply also to the case of %
Disjunctive Normal Form (DNF) formulas with at most $p(n)$
cubes.
This is a corollary of the known duality between backward- and forward-reachability:
The dual transition system is
$
	\dual{(\Init,\tr,\Bad)} \eqdef (\Bad, \tr^{-1}, \Init)
$
where $\tr^{-1}$ is the inverse (between pre- and post-states) of the transition relation, obtained from $\tr$ by switching the roles of $\voc$ and $\voc'$.
The dual of a class of transition systems is the class of dual transition systems.
The dual of a formula is $\dual{\varphi} \eqdef \neg \varphi$. The dual of a class of formulas is the class of dual formulas.
We have that $I$ is an inductive invariant w.r.t.\ $(\Init,\tr,\Bad)$ iff $\dual{I}$ is an inductive invariant w.r.t.\ $\dual{(\Init,\tr,\Bad)}$.
For an algorithm $\A$, The \emph{dual algorithm} $\dual{\A}$ is the algorithm that, given as input $(\Init,\tr,\Bad)$, executes $\A$ on $\dual{(\Init,\tr,\Bad)}$ and returns the dual invariant. By applying the dual algorithm to the dual transition systems and target invariants, we obtain:
\begin{lemma}
Polynomial-length inference of $\progs{}$ w.r.t.\ $\cnf{p(n)}{}$ has the same complexity as polynomial-length inference of $\dual{\progs{}}$ w.r.t.\ $\dnf{p(n)}{}$.
\end{lemma}
\section{Polynomial-Length Transition Systems and Hoare Information Complexity}
\label{sec:poly-tr}
In our definitions, the complexity of a black-box query-based inference algorithm (\Cref{def:query-complexity}) is a function of the target invariant length derived from the vocabulary size, and does not depend on the length of the representation of the transition relation, which the inference algorithm cannot access directly.
Accordingly, the general case in our lower bound on the query complexity of Hoare-query algorithms is the class of \emph{all} transition systems, and we show that to solve the general case the number of Hoare queries must be exponential in $n=\card{\voc}$.
\sharon{I thought that the following interrupts the flow and can be omitted. Maybe when this section appears after talking about concept learning it makes more sense, not sure} \yotam{it's here to justify that our results are interesting, so I'd like to have it somewhere, somewhere else is OK}%
We emphasize that this result holds despite the reasonably-sized class of target invariants (an exponential number, not doubly exponential), and that Hoare queries are rich (can be used with any precondition and postcondition).

Can inference algorithms utilize an assumption that the transition relation can be expressed by formulas of polynomial size? Formally, this asks for an analysis of the query complexity as it depends also on $\card{\tr}$ in addition to $\card{\voc}$.
An important technical difference is that in this setting an algorithm may attempt to ``breach'' our black-box definition, and (concept-) learn the transition relation formula itself; once it is obtained the algorithm can deduce from it whether an invariant exists using unlimited computational power (see \Cref{sec:gen-intro}).
The possibility of such concept learning reflects on the statement of \Cref{thm:hoare-nopoly} when the complexity definition is altered.
With unlimited computational power and exponentially-long queries, the algorithm can learn the formula of $\tr$:
\begin{lemma}
\label{lem:lem-tr-length-unrestricted-poly}
There exists a computationally unrestricted Hoare-query inference algorithm $\bbalg{\horaclesym}{}$ with query complexity polynomial in $\card{\voc},\card{\TS}$ for the class of transition systems $\hardclass$ (\Cref{sec:hard-class}). %
\end{lemma}
\begin{proof}
Assume we are attempting to infer with Hoare queries an invariant for $\fts{\phi} \in \hardclass^k$, with $\card{\phi} \leq m$. We show how to identify $\phi$ in a number of Hoare queries polynomial in $m$, from which using unlimited computational power the algorithm can check if an invariant for $\fts{\phi}$ exists.

Use the halving algorithm/majority vote~\cite{barzdins1972prediction,DBLP:journals/ml/Angluin87,DBLP:journals/ml/Littlestone87} to learn $\phi$. For completeness, we describe it here:
At each step, consider $S$, the set of formulas (over the same vocabulary as $\phi)$ of length at most $m$ that are consistent with the results of the queries performed so far (namely, $\theta \in S$ if performing the same Hoare queries on $\fts{\theta}$ yields the same results as were observed). We can terminate if $S$ includes only one formula.
Otherwise, take $\hat{\phi} \in S$ to be the formula that is $\true$ on a valuation $v$ iff the majority of formulas in $S$ so far are $\true$ on $v$. Check whether $\phi \equiv \hat{\phi}$ using a Hoare query to be described below. If $\phi \equiv \hat{\phi}$ we are done, otherwise we obtain a counterexample: a valuation on which $\phi,\hat{\phi}$ disagree. In the next step, the set $S$ is reduced by at least a factor of two, so the process terminates after a number of iterations polynomial in $m$.

It remains to implement using Hoare queries the check of whether $\phi \equiv \hat{\phi}$ and obtaining a counterexample.\footnote{
	\Cref{cor:monmax-lower} implies that inductiveness queries alone cannot perform this check, because there the transition relation formulas are also polynomial in $n$, as we discuss below concerning the results of \Cref{thm:hoare-inductive-main-theorem}. \sharon{where is "below?} \yotam{in this appendix, don't have a ref, added ``concerning''}
}
This can be done using the query $\horacle{\phi}{\alpha,\beta}$ where:
$\alpha$ is the formula $\neg a \land \neg b \land \neg e$ %
(with any $\vec{y},\vec{x}$),
and $\beta$ is the post-image of $\fts{\hat{\phi}}$ on $\alpha$.
\yotam{was informal and you said confusing, rephrased}If the result is $\true$, the equivalence query returns $\true$.
Otherwise, the valuation differentiating $\phi,\hat{\phi}$ is obtained from the pre-state (in the propositions $\vec{y},\vec{x}$) of the counterexample to the Hoare query (using \Cref{lem:hoare-cti}). This is correct because $\psigmatr{\phi},\psigmatr{\hat{\phi}}$ have different transitions from the same state $\sigma$ iff $\phi(\vec{y},\vec{x}) \neq \hat{\phi}(\vec{y},\vec{x})$ where ($\vec{y},\vec{x}$ are from the interpretation of these propositions in $\sigma$.
\end{proof}

As~\citet{DBLP:journals/ml/Angluin87} recognizes, this result relies on queries on candidates that are exponentially-long.
When queries can be performed only on polynomially-long formulas and their choice can use unrestricted computational power, the question of whether a result analogous to \Cref{thm:query-nopoly} exists is related to open questions in concept learning, such as\sharon{not convinced why this is good} \yotam{slightly repharsed. otherwise, just say it's related to open questions without details?} whether a polynomial number of equivalence and membership queries can identify a formula of length at most $m$~\cite{DBLP:journals/jcss/BshoutyCGKT96}\yotam{is this actually true? old reference}.
We thus propose the following conjecture, a variant of \Cref{thm:query-nopoly} where the complexity depends also on the size of the transition relation:
\begin{conjecture}
\label{conj:lem-tr-length-length-exp} %
Every Hoare-query inference algorithm $\bbalg{\horaclesym}{}$, \emph{even computationally-unrestricted}, \emph{querying on formulas polynomial in $\card{\voc}+\card{\TS}$}, for the class of transition systems $\hardclass$ (\Cref{sec:hard-class}) and %
and for any class of target invariants $\mathcal{L}$ s.t.\ $\moncnf{n}{} \subseteq \mathcal{L}$,
has query complexity superpolynomial in $\card{\voc}+\card{\TS}$.
\end{conjecture}

If \Cref{conj:lem-tr-length-length-exp} is true, a result analogous to \Cref{thm:hoare-nopoly} can be obtained, obtaining superpolynomial lower bounds not only in $\card{\voc}$ but also in $\card{\TS}$.

\sharon{relate this to equiv only?} We emphasize that the exponential lower bound we obtain in \Cref{thm:hoare-inductive-main-theorem} is already exponential also in $\card{\TS}$ (and this holds even when candidates can be exponentially long), %
as the transition relations in $\monmax$ are all of size polynomial in their vocabulary.
\sharon{too early here} \yotam{I guess it's OK now that it's in the appendix}
 \fi

\end{document}